\documentclass[superscriptaddress,floatfix,letterpaper,accepted=2021-09-14]{quantumarticle}
\pdfoutput=1

\usepackage{float}
\usepackage[export]{adjustbox}
\usepackage{amsmath,mathtools,dsfont,amssymb}
\usepackage{braket}
\usepackage{enumitem}
\usepackage{amsfonts}
\usepackage{graphicx}

\usepackage{hyphenat}
\usepackage{algpseudocode}
\usepackage{algorithmicx}
\usepackage[ruled,vlined]{algorithm2e}
\usepackage{microtype}
\usepackage[caption=false]{subfig}
\usepackage[svgnames]{xcolor}
\usepackage{tabularx}
\usepackage{booktabs}
\usepackage[utf8]{inputenc}
\usepackage{bm}
\usepackage{siunitx}
\usepackage[hang,flushmargin]{footmisc}

\usepackage{amsthm}

\usepackage[numbers,sort&compress]{natbib}
\usepackage{pdfpages}
\makeatletter
\AtBeginDocument{\let\LS@rot\@undefined}
\makeatother
\usepackage{afterpage}

\usepackage{hyperref}
\hypersetup{colorlinks,linkcolor={DarkBlue},citecolor={DarkBlue},urlcolor={DarkBlue}}
\usepackage[capitalize]{cleveref}
\usepackage{url}
\newtheorem{theorem}{Theorem}[section]

\newtheorem{lemma}[theorem]{Lemma}

\begin{document}
	
	\title{Exact holographic tensor networks for the  Motzkin spin chain}
	
	\author{Rafael N. Alexander}
	\affiliation{Centre for Quantum Computation and Communication Technology,
		School of Science, RMIT University, Melbourne, VIC 3000, Australia}
	\affiliation{Center for Quantum Information and Control, University of New Mexico, Albuquerque, NM 87131, USA}
	\email{rafael.alexander3@rmit.edu.au }
	\author{Glen Evenbly} 
	\affiliation{School of Physics, Georgia Institute of Technology, Atlanta, GA 30332, USA}
	\email{gevenbly3@gatech.edu}
	\author{Israel Klich} 
	\affiliation{Department of Physics, University of Virginia, Charlottesville, Virginia 22903, USA}
	\email{klich@virginia.edu}
	\begin{abstract}
		The study of low-dimensional quantum systems has proven to be a particularly fertile field for discovering novel types of quantum matter. 
		When studied numerically, low-energy states of low-dimensional quantum systems are often approximated via a tensor-network description. The tensor network's utility in studying short range correlated states in 1D have been thoroughly investigated, with numerous examples where the treatment is essentially exact. Yet, despite the large number of works investigating these networks and their relations to physical models, examples of exact correspondence between the ground state of a quantum critical system and an appropriate scale-invariant tensor network have eluded us so far. Here we show that the features of the quantum-critical Motzkin model can be faithfully captured by an analytic tensor network that exactly represents the ground state of the physical Hamiltonian. In particular, our network offers a two-dimensional representation of this state by a correspondence between walks and a type of tiling of a square lattice. 
	\end{abstract}
	
	\maketitle
	
	\section{Introduction}
	One of the hallmarks of critical behavior is the divergence of correlations, and the emergence of scale invariance; the low-energy behavior of the system seems to be of a similar nature on small and large scales. Beyond the physical beauty of such states, their treatment has helped develop many important ideas and tools, such as, conformal field theory (CFT) and the renormalization group (RG) \cite{wilson1983renormalization,fisher1998renormalization}. 
	
	At temperatures approaching absolute zero, criticality is of a quantum nature and is a focal point of interest for low-temperature many-body physics.  At such quantum critical points, long-range quantum fluctuations and correlations allow the system to sustain emergent large-scale quantum behavior (see e.g. Ref.~\cite{sachdev2011quantum}). 
	Such phenomena have been observed in many experiments, including in magnetic systems where quantum critical behavior may develop when a magnetic transition is driven by changes in chemical doping~\cite{lohneysen1994non}, pressure~\cite{julian1996normal}, external fields~\cite{grigera2001magnetic} and other system parameters.

	A major current challenge is therefore the detailed description of scale-invariant ground states of quantum systems. %
	This is a particularly difficult problem since such systems are typically highly entangled, and harder to study than gapped generic ground states with short-range correlations.
	
	A recent promising approach to simulating many-body states is via so-called tensor networks. Tensor network notation~\cite{Orus2014,bridgeman2017hand} offers a convenient graphical representation of the entanglement structure of many-body quantum states.  A particularly simple class of 1D tensor networks, particularly useful for describing spin-chains with finite range correlations, are known as {\it matrix product states} (MPS). These were introduced in Ref.~\cite{fannes1992finitely}, and are the variational class of states used in White's {\it density matrix renormalization group} (DMRG) numerical procedure ~\cite{White1992}, arguably the most successful tool for numerical investigation of quantum phases in 1D. Another class of tensor network states, known as the {\it multi-scale entanglement renormalization ansatz} (MERA) \cite{Vidal2008MERA}, was proposed by Vidal~\cite{Vidal2007ER} to be specially tailored for describing quantum critical points.

	While there are a number of examples of physical spin systems for which an exact matrix product state representation is known for the ground state, most notably the celebrated AKLT state of Affleck, Lieb, Kennedy and Tasaki \cite{affleck1987rigorous}, the situation for critical systems is quite different. To our knowledge, until now, there has been no exact example of a scale-invariant tensor network that represents the ground state of a simple local Hamiltonian. Such networks are typically computed numerically. The only prior case of an analytical MERA we are aware of describes a sequence of approximate descriptions for free fermions, shown in Ref.~\cite{Evenbly2016, haegeman2018rigorous}.

	Here we provide the first {\it exact} analytic hierarchical tensor network for describing a critical state: the ground state of the Motzkin spin chain. Our results suggest that scale-invariant tensor networks may be useful beyond dealing with CFTs. Indeed, it is well known that the Motzkin Hamiltonian does not have the energy level scaling associated with a CFT~\cite{Movassagh2016Power}.  Despite this, the Motzkin spin chain bears many hallmarks of quantum critical systems. For instance, the Hamiltonian's gap closes polynomially quickly~\cite{bravyi2012criticality}, the ground state has logarithmically growing entanglement entropy~\cite{bravyi2012criticality}, and the ground state has a scale invariant description. This last point is one of our main results.

	The Motzkin model, as well as the closely related Fredkin model, grew out of the study of ``frustration-free'' Hamiltonians.  These are an important class of Hamiltonians, where the ground state is also a simultaneous ground state for the local interactions comprising the Hamiltonian. Frustration-free Hamiltonians have been recently used for constructing novel quantum states of matter and also as representations for a variety of quantum optimization problems. Examples of frustration-free Hamiltonians include Kitaev's toric code \cite{kitaev2003fault}, the AKLT Hamiltonian, and the Rokhsar-Kiveson model for a quantum dimer gas \cite{rokhsar1988superconductivity}. Moreover,  any MPS state is the ground state of a frustration-free Hamiltonian~\cite{fannes1992finitely, perez2007peps, schuch2010peps, schuch2011classifying}.

	The Motzkin Hamiltonian, a spin-1 Hamiltonian introduced by Bravyi et al. in Ref.~\cite{bravyi2012criticality}, represents an important new class of Hamiltonians that are both frustration-free and critical. The model admits a straightforward geometric interpretation in terms of random walks called ``Motzkin walks'' and its ground state is exactly solvable.  Moreover it is the starting point for several generalizations that uncovered rich new possibilities for how the entanglement in a ground state can scale. In particular, Movassagh and Shor \cite{Movassagh2016Power} showed how a higher spin (``colored'') version of the model can feature much enhanced entanglement (going from typical logarithmic behavior in critical spin chains to a power law). Motivated by this model, Zhang, Ahmadain and Klich~\cite{Zhang2017novel} found a parametric deformation of the model that yielded a continuous family of frustration-free Hamiltonians featuring a new quantum phase transition interpolating between an area-law phase and a ``rainbow'' phase with volume scaling of half-chain entanglement entropy. A spin 1/2 version of the model has been proposed based on so-called Fredkin gates by Salberger and Korepin~\cite{salberger2017entangled,dell2016violation}, and its deformation was presented in Ref.~\cite{salberger2017deformed,zhang2017entropy,udagawa2017finite}. An interesting recent variation of this class of models can be found in, e.g. Ref.~\cite{sugino2018area} using symmetric inverse semigroups.
	
	Here we present two types of tensor network that faithfully capture both the geometric properties and the entanglement structure of the colorless models: The binary height network and the height renormalization network. As described below, each path is mapped onto a tiling of a grid that defines the tensor network. Thus, our networks have the rules that govern the random walker baked into their building blocks. In particular, the second network we propose has a MERA-like structure~\cite{Bal2016} and defines a natural renormalization process of Motzkin walk configurations. 
	\section{Motzkin spin chains}
	As stated above, the spin-1 Motzkin model was introduced as an example of a critical spin chain described by a frustration-free Hamiltonian with nearest-neighbor interactions. These have been generalized to include a deformation parameter $t$, where $t=1$ corresponds to a critical point.
	
	The Motzkin Hamiltonian has a unique zero-energy frustration-free ground state for any value of $t>0$.
	For our present purpose, we need the exact form of the ground state as we describe below. (The interested reader may refer to the detailed description of the Hamiltonian in the Appendix~\ref{app:A} and in the references). 
	This ground state is a superposition of walks called ``Motzkin Walks''. A Motzkin walk $w$ is a walk on the $\mathbb{Z}^{2}$ lattice using the line segments $\{ \diagup,  \text{\textbf{---}}, \diagdown \}$ that start at $(0,0)$, go to $(2n,0)$, and never go below the $x$ axis. Each walk represents a spin configuration via identifying the Motzkin line segments with the  local spin $\{ \ket{1}, \ket{0}, \ket{-1}\}$ states, respectively.  The ground state can be written as \cite{Zhang2017novel}:
	
	\begin{equation} |\text{GS}\rangle = \frac{1}{\mathcal{N}}\sum_{\substack{w \in \{ \text{Motzkin}\ \text{walks}\}}} t^{\mathcal{A}(w)} |w\rangle \label{cgs}. \end{equation} 
	Here $\mathcal{A}(w)$ denotes the area below the Motzkin walk $w$, and $\mathcal{N}$ is a normalization factor. A similar type of ground state occurs in the Fredkin model, which is a half-integer spin model with essentially the same structure; it only lacks  the ``flat'' move.
	
	The half-chain entanglement entropy, which is a measure of the degree of quantum correlations in the system, is maximal for the $t=1$ case, where it grows logarithmically in $n$. For $t<1$ and $t>1$, the ground state satisfies an area law ~\cite{zhang2017entropy}: the entanglement entropy is bounded by a constant independent of system size. The deformed Motzkin and Fredkin walks can naturally be viewed as constrained trajectories of a random walker in the presence of drift; the ``x'' axis plays the role of time; and the $t$ parameter measures the strength and direction of the drift. More details can be found in references \cite{salberger2017entangled,salberger2017deformed,Zhang2017novel,zhang2017entropy}.
	
	Though we will focus on the unweighted $t=1$ version of the ground state, we will also describe the details of the $t>0$ cases (see Eqs.~(\ref{eq:binheightweight}) and (\ref{eq:rnweight})). 
	

	The Motzkin model can also be studied with periodic boundary conditions~\cite{Movassagh2016Power}. The ground space of the periodic model has a $4n+1$ degeneracy, and the ground states $\ket{\Psi_{k}}$ consist of equal superpositions over all spin-$z$ configurations that have a total magnetization $k$, where $-2n\leq k \leq 2n$~\cite{Movassagh2016Power}. Recently, such states were shown to act as approximate quantum error correcting codes~\cite{Brandao2018}. 
	
	\section{Binary height tensor network}
	The Motzkin model can be described as a ``height model'': the height of each walk is encoded within a field value at each point.  This is a natural starting point for constructing a field theory description of the colorless Motzkin state, as was done in Ref.~\cite{Chen2017MotzkinLifshitz}. 

	If we describe the state using a binary representation of the heights, we see that, since at site $x$, the walk could have reached at most height $x$, we need at most $\log_{2}(2x)$ bits to encode the height. The height encoded at site $x+1$ results from the adding (or subtracting) the spin value at site $x+1$. 
	Thus, our first step is to generate a tensor network that implements these additions. We find it convenient to encode both the walk and the binary addition using a set of tiles, as explained below.
	
	\subsection{Walks as tiles}
	Consider the square tiles $A_{1},\dots,A_{6}$ represented below:%
	\begin{align}
		\includegraphics[width=0.8\linewidth]{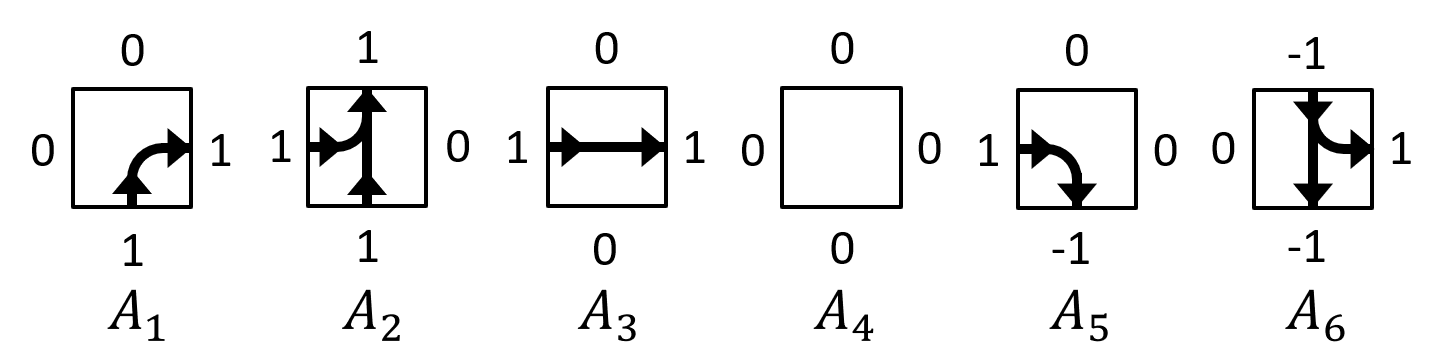}\label{eq:tiles1}
	\end{align}
	Top and bottom edges are labeled by $1$, $-1$, or $0$ if that edge touches a ``$\uparrow$''-line, a ``$\downarrow$''-line, or neither, respectively. Left and right edges are labeled by $1$ or $0$ depending on if that edge touches a ``$\rightarrow$''-line or not.  
	
	Given a tiling of a square grid, we say that a tiling is {\it valid} if all edges match. To represent a $2n$-step Motzkin walk by a valid tiling, we consider a square-grided ``step pyramid'', as shown in Fig.~\ref{fig:t1Exact}. The steps vary in length because only a logarithmically growing number of bits ($\lfloor \text{log}_{2} 2x\rfloor$) is required to store the height at the $x^{\text{th}}$ column, for $1\leq x\leq n$. Also, the network is symmetric about the halfway point. We impose boundary conditions on the exterior north, east, and west edges such that they are all equal to zero. The values on the south boundary correspond to the local spin values of the walk. As shown in Fig.~\ref{fig:t1Exact}, the height of the walk is encoded in the binary strings between columns of tiles. This defines an isomorphism between Motzkin walks and valid tilings.
	
	\begin{figure}
		\begin{center}
			\includegraphics[width=\linewidth]{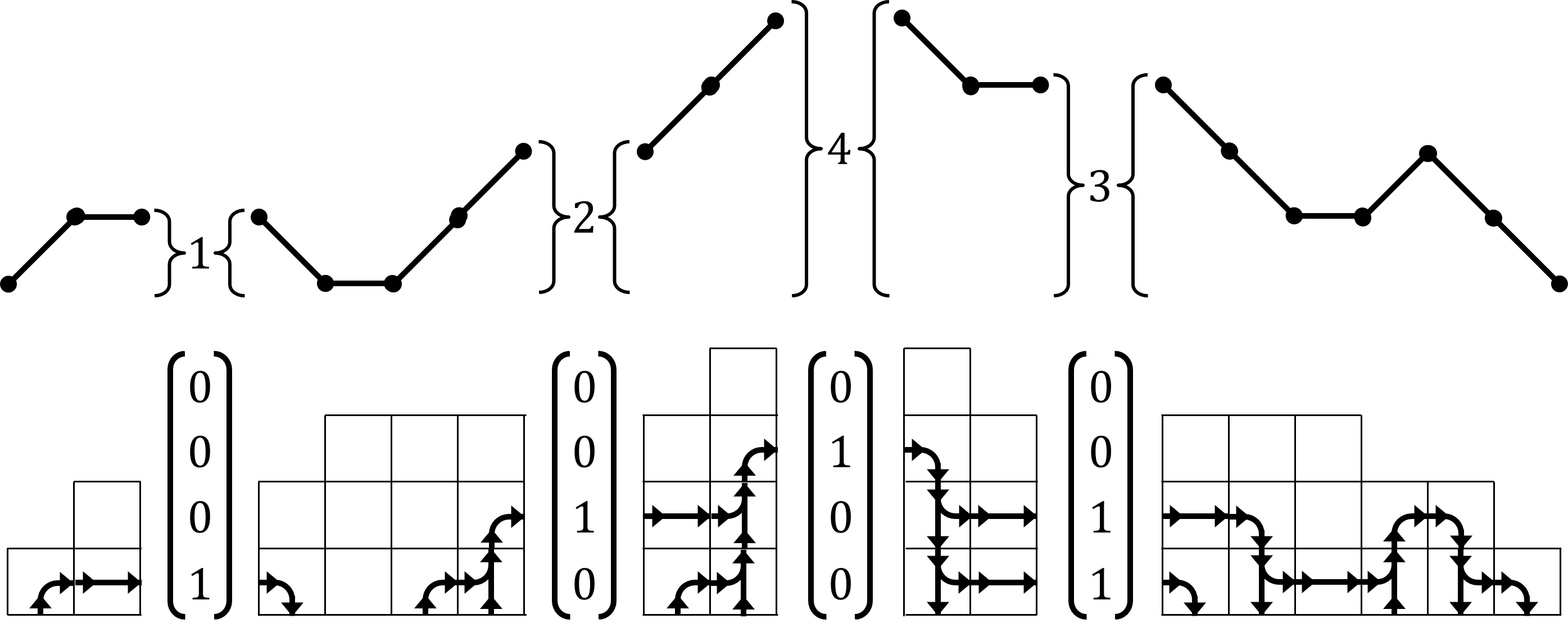}
		\end{center}
		\caption{Tile representation (bottom) for a Motzkin walk (top). The tiles from Eq.~(\ref{eq:tiles1}) are placed on a square-gridded step pyramid. The height of the walk is encoded in binary between columns of tiles, where ``1'' and ``0'' values are associated with the presence or absence of a horizontal arrow. For example, the middle point has height $4$, which, in the network is encoded by the sequence $0100$, representing $4$ in binary.  Vertical arrows act as ``carry'' bits in the binary addition process.  Each valid tiling can be described pictorially. A value of 1 at the base of the pyramid starts an arrowed line that travels through the bulk. These arrowed paths travel horizontally to the right, but 
			to move vertically, two arrowed paths must fuse, and  similarly, an arrowed path bifurcates when moving downwards.} \label{fig:t1Exact} 
	\end{figure}

	\subsection{Tiles as tensors}
	Let us define a tensor network that sums over all configurations along the bottom edge of a square-lattice step pyramid that yield a valid tiling.
	First, we identify each tile $A_{j}$ with a tensor. We define the rank-one four-index tensor
	\begin{align}
		\delta_{k_{1}k_{2}k_{3}k_{4} }(w, x, y, z) = \delta_{k_{1}, w} \delta_{k_{2}, x}\delta_{k_{3}, y}\delta_{k_{4}, z}. \label{eq:delta1}
	\end{align}
	where $\delta_{j, k}$ is the Kronecker delta function and $w, x, y, z \in\mathbb{Z}$. Next, identify each tile in Eq. \eqref{eq:tiles1} with such a tensor where the values $w, x, y, z$ correspond to the north, east, south, and west boundary values, respectively. Explicitly, 
	\begin{align*}
		A_{1;\vec{k}} &= \delta_{\vec{k}}(0, 1, 1, 0) ~~;~~ 
		A_{2;\vec{k}} = \delta_{\vec{k}}(1, 0, 1, 1) \\  
		A_{3;\vec{k}} &= \delta_{\vec{k}}(0, 1, 0, 1) ~~;~~  
		A_{4;\vec{k}} = \delta_{\vec{k}}(0, 0, 0, 0) \\  
		A_{5;\vec{k}} &= \delta_{\vec{k}}(0, 0, -1, 1) ~~;~~
		A_{6;\vec{k}} = \delta_{\vec{k}}(-1, 1, -1, 0) 
	\end{align*}
	where we compress the four index notation $k_{1},k_{2},k_{3},k_{4},$ into a 4-tuple $\vec{k}$. Now we can define
	\begin{align}
		\begin{tikzpicture}
			\begin{scope} 
				\node {$\displaystyle B_{\vec{k}}= \sum^{6}_{j=1} A_{j;\vec{k}} =$};
			\end{scope}
			\begin{scope}[xshift = 2.5cm]
				\node {\includegraphics[width=0.25\linewidth]{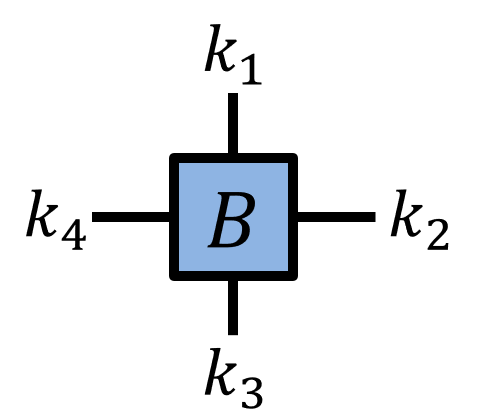}};
			\end{scope}
		\end{tikzpicture} \label{eq:TNBdef}
	\end{align}
	which is the primary building block of our first tensor network representation of the Motzkin ground state, which we call the {\it binary-height tensor network}. It is shown in Fig.~{\ref{fig:TN2defs}(a)}. Recall that the north, east, and west facing outer edges of each tiling must be projected onto the value $0$. This boundary condition is set by indices on these edges with 
	\begin{align}
		\begin{tikzpicture}
			\begin{scope} 
				\node {$\displaystyle \delta_{k_{1}, 0} =$};
			\end{scope}
			\begin{scope}[xshift = 1cm]
				\node {\includegraphics[width=0.05\linewidth]{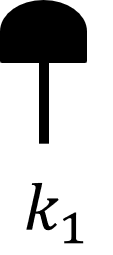}};
			\end{scope}
		\end{tikzpicture} \label{eq:bdtensdef}
	\end{align}

	

	
	\begin{figure}
		\includegraphics[width=\linewidth]{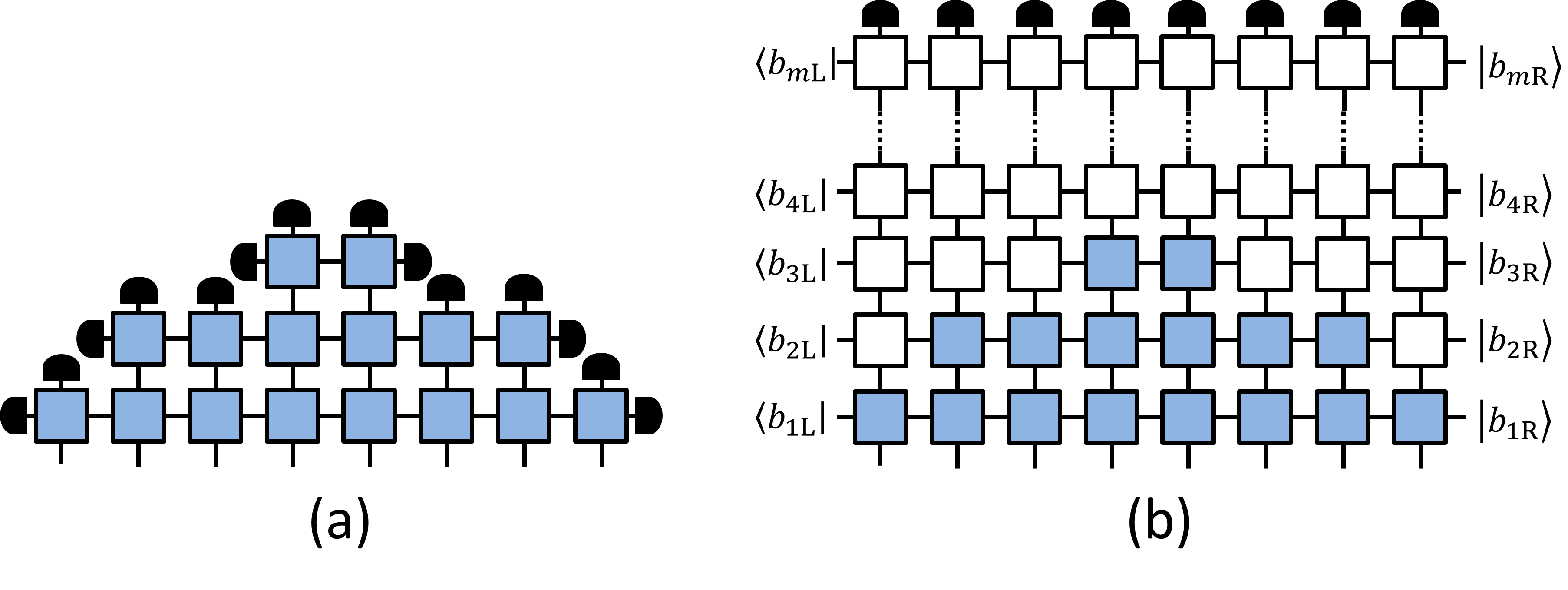}
		\caption{$\mathbf{(a)}$ 8-spin example of the binary height tensor network.  Physical indices for this network are located at the bottom. This network requires $O(n \log n)$ square tensors. $\mathbf{(b)}$ Generalized binary-height tensor network for summing over all walks that start at height $p$ and end at height $q$. The left and right boundaries are contractions with various $\ket{b_{i}}$, where $b_{i}\in \{0, 1\}$, and  $b_{m\text{L}}\dots b_{1\text{L}}$ and $b_{m\text{R}}\dots b_{1\text{R}}$ are the binary expansions of $p$ and $q$, respectively. In the case where $p=q=0$, this is equivalent to the original network shown in (a). The original $B$ tensors (colored blue) have been padded from above with identical $B$ tensors (colored white).} \label{fig:TN2defs}
	\end{figure}
	
	Contracting this network gives a value of $1$ if the tiling is valid and $0$ otherwise. Therefore, it represents an equal-weight superposition of Motzkin walks, which is exactly the $t=1$ Motzkin ground state. 
	
	The height of this network is $\lfloor \log_{2} n \rfloor$, therefore the bond dimension between the two halves of the network grows as $O(n)$. The total number of $B$ tensors required scales as $O(n \log{n})$~\footnote{More specifically, for  $n = 2^{y}-1$, where $y$ is a positive integer, exactly $2(n+1)\log_{2}(n+1) - 2n$ many $B$ tensors are required.}. An immediate upper bound on the entanglement entropy between the two sides, as estimated by the number of cuts needed, is simply $\log{n}$.
	
	The generalization to the spin-1/2 case (known as the Fredkin model) and the case of area-weighted walks is given in Eqs.~(\ref{eq:fredbintens}) and (\ref{eq:binheightweight}), respectively.
	Below we show how the binary height network can be straightforwardly generalized to describe the spin $1/2$ Fredkin model, as well as the  Motzkin model with periodic or height-varied boundary conditions.


	\subsection{Boundary conditions}\label{sec:TN2bound}
	
	First, we consider modifying the shape of the network. If we extend the width-$2n$ square-gridded step pyramid of the binary height-network to an $m \times 2n$ square-gridded rectangle with $m\geq \lfloor\log_{2} 2n \rfloor$, then there is no increase in the number of valid tilings. In fact, tilings of the rectangle are merely embedded tilings of the pyramid that have been  padded from above with blank $A_{4}$ tiles. 
	
	Next, consider generalizing the left and right boundary conditions of the binary-height network (extended to a height $m$ rectangle) from $\ket{0}^{\otimes \lfloor \log_{2} 2n \rfloor}$ to a product of $\ket{0}$ and $\ket{1}$ states
	\begin{align}
		\ket{\vec{b}_{\text{L(R)}}} = \bigotimes^{m}_{k=1} \ket{b_{k\text{L(R)}}}. \label{eq:bdvec}
	\end{align}
	where $b_{k\text{L(R)}}$ is the $k^{\text{th}}$ bit in the binary expansion of the integer $p$  ($q$) between 0 and $2^{m}-1$. 
	
	Fig.~{\ref{fig:TN2defs}(b)} shows a tensor network that incorporates these changes. If $2^{m}-1\geq \text{max}(p,q) +n$, this network represents the equally weighted sum over all walk configurations with non-negative heights starting at height $p$ and finishing at height $q$.

	If the left and right height boundary conditions are increased so that $\text{min}(p,q)\geq n$, then all length $2n$ walks with net height $q-p$ will be generated by the network. This is because no walk is long enough to reach a negative height value. This observation can be used to find a representation of $\ket{\Psi_{k}}$. Choosing left and right boundary vectors to represent any $p$ and $q$ such that $q-p=k$ and $\text{min}(p,q)\geq n$ will generate $\ket{\Psi_{k}}$. One example is to set $\text{max}(p, q) = 2 n$ and $\text{min}(p, q) = 2 n - \lvert k \rvert$.


	\section{Zipping up: an exact network RG transformation}\label{sec:RNTN} 
	The binary-height tensor network discussed above offers a compact and intuitive representation of the Motzkin ground state. Nevertheless, the $O(n \log{n})$ number of $B$ tensors required is suboptimal. Here we show that the ${t=1}$ Motzkin and Fredkin ground states, and their generalizations with periodic boundary conditions, can be represented exactly using a tensor network that only requires $O(n)$ many tensors.
	
	We offer two equivalent constructions. In the main text we present a tile-based approach consistent with the construction of the binary-height tensor network presented above. In Figs.~\ref{fig:TensorDef}-\ref{fig:Compare} we provide an independent construction that is more closely related to existing tensor-network methods. In particular, it leverages the framework of U(1)-invariant tensors, originally described in Ref.~\cite{singh2011tensor}.
	
	The key idea behind the tile-based construction is a network renormalization group step associated to the ``zipper lemma'':
	\begin{align}
		\includegraphics[width=0.4 \linewidth]{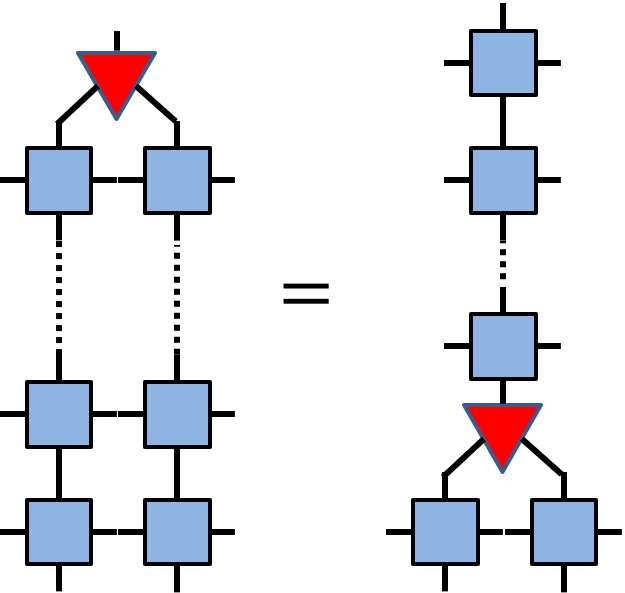}\label{eq:zipperlemmaT}
	\end{align}
	where the triangular tensor will be defined below, and the proof is given in the appendix. Sequential application of the zipper lemma leads to a new tensor network where each layer discards the highest frequency components of the layer below. Thus, it naturally represents a renormalization process for Motzkin walks, and resembles a MERA.
	
	First we will treat the periodic case, which is more closely related to the binary-height tensor network and involves a simpler tile set. After this, we will generalize the construction to the original Motzkin model. 
	
	\subsection{The height renormalization network with periodic boundary conditions}
	
	We will require the following three index tensor:
	\begin{align}
		\begin{tikzpicture}
			\begin{scope} 
				\node {$\displaystyle T = \sum_{i, j, k=-1}^{1} t_{i j k} \ket{i}\ket{j}\bra{k} =$};
			\end{scope}
			\begin{scope}[xshift = 3cm]
				\node {\includegraphics[width=0.15\linewidth]{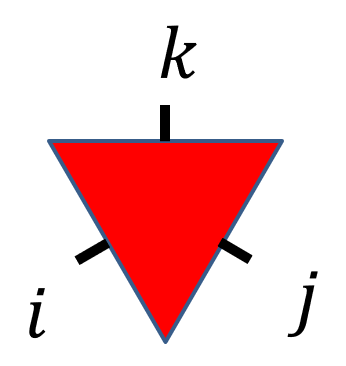}};
			\end{scope}
		\end{tikzpicture} \label{eq:triangletensor}
	\end{align}
	where,
	\begin{align}
		t_{ijk} \coloneqq \delta_{i+j, k}
	\end{align}
	and all indices $i, j$, and $k$ are restricted to values $\{-1, 0, 1\}$. Consequently, the values of $i$ and $j$ cannot be both 1 or both -1. 
	
	Now we introduce the basic unit of renormalization. Consider the tensor network shown below, which maps two spin values $s_1$ and $s_2$, to a single spin value $s_3$
	\begin{align}
		\includegraphics[width=0.25\linewidth]{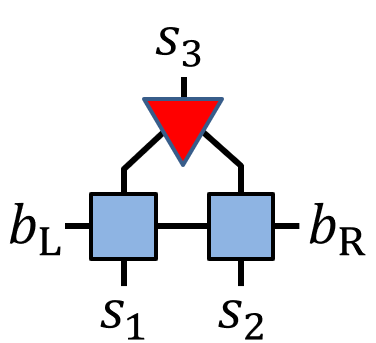}\label{eq:renormunit}
	\end{align}  
	where
	\begin{align}
		s_{3} =  \begin{cases}
			\lfloor (s_{1}+s_{2})/2\rfloor,& \text{if } b_{\text{L}} = 0\\
			\lceil (s_{1}+s_{2})/2\rceil,              & \text{if } b_{\text{R}} = 1
		\end{cases}
	\end{align}
	and
	\begin{align}
		b_{\text{R}}=b_{\text{L}} + s_{1} +s_{2} \quad (\text{mod}\, 2).
	\end{align}
	This map can be represented pictorially, as shown in Fig.~{\ref{fig:TNRNtable2}(a)}. 
	
	\begin{figure}
		\begin{center}
			\includegraphics[width=\linewidth]{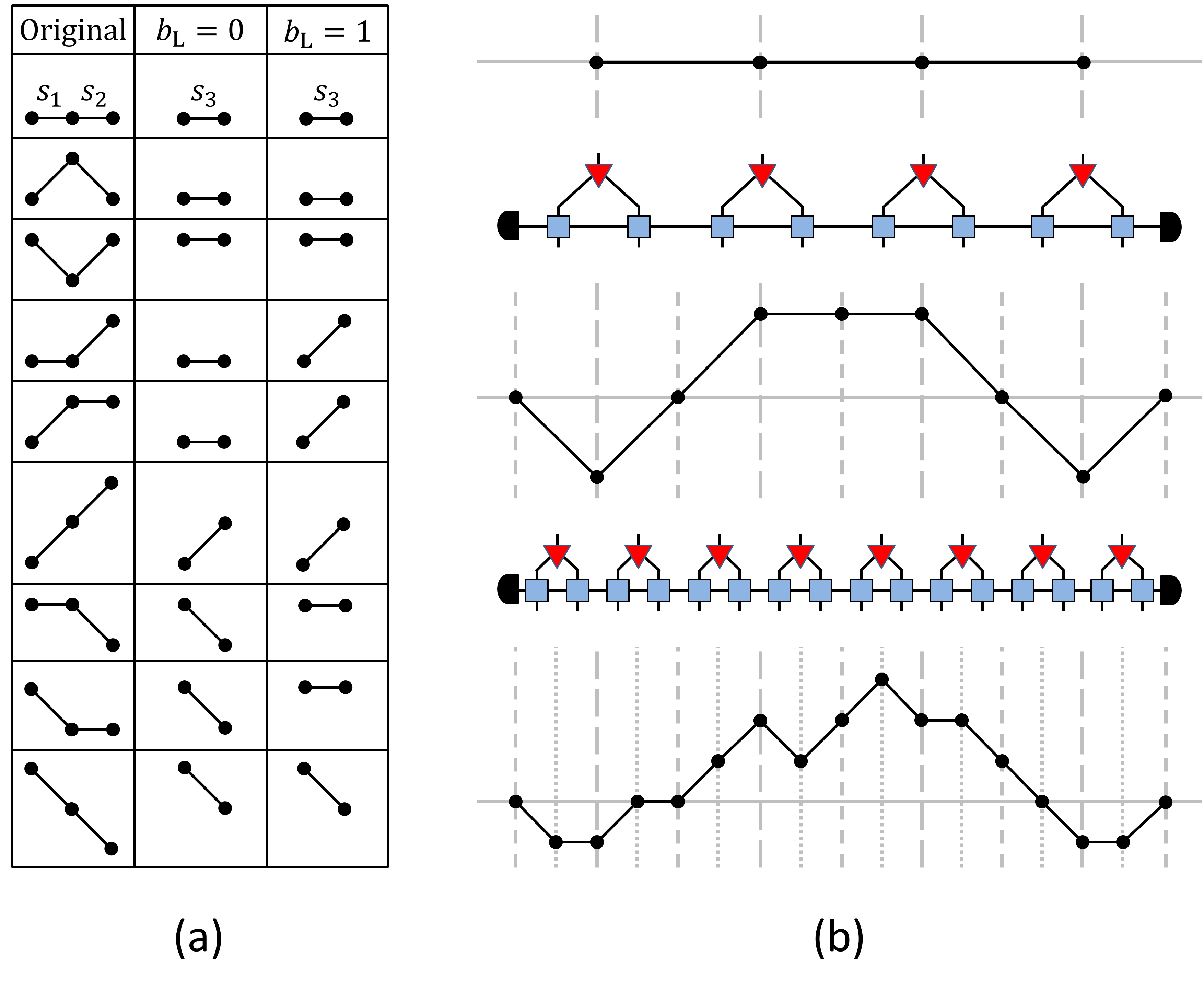}
		\end{center}
		\caption{$\mathbf{(a)}$ Table showing      the map applied by the tensor network in Eq.~(\ref{eq:renormunit}) between the top and bottom indices. The value of $s_{3}$ depends on whether the left index value $b_{\text{L}}$ is $0$ or $1$. Equivalently, it depends on if the segment corresponding to $s_{1}$ starts at an even or odd height, respectively. Note that the values on the right two columns do not depend on the height of the middle node in the first column. $\mathbf{(b)}$ Example of walk renormalization. The height achieved at each dashed line on the top walk is equal to half the height (rounding down) at the same point in the walk below. As was pointed out in (a), the height values at the thickly dashed lines for the higher walks do not depend on the heights at the more finely dashed lines for walks below. This is the information that is discarded by the coarse graining process.}\label{fig:TNRNtable2}
	\end{figure}
	
	Next, consider the action of a single layer of these 3-tensor renormalization networks, as shown in Fig.~{\ref{fig:TNRNtable2}(b)}. As a linear map acting on a walk configuration on the $2n$ bottom indices, it outputs a ``coarse grained'' version of that walk on the $n$ top indices. More specifically, the height of the walk between sites $l$ and $l+1$ on the top indices is half the height (rounding down) of that between sites $2l$ and $2l+1$ at the bottom indices.

	If this course-graining procedure were repeated $\lfloor \log_{2} (2n) \rfloor$ times, it could be represented by the network shown in Fig.~{\ref{fig:rntndefs}(a)}. It requires $O(n)$ many square and triangular tensors~\footnote{More specifically, for $2n=2^{y}$ for some integer $y>0$, exactly $4n (1-2^{-2n})$ many $B$ tensors are required, and  $2n (1-2^{-n})$ many $T$ tensors are required. }. Note that we have left bit values of the left and right boundary vectors $\ket{\vec{b}_{\text{L}}}$ and $\ket{\vec{b}_{\text{R}}}$ unspecified. These encode two integers $p$ and $q$, respectively. 
	
	Contracting this tensor network with a walk configuration gives the value 1 if the walk reaches a net height of $p-q$, and gives zero otherwise. Therefore, by choosing $\text{max}(p,q)=2n$ and $\text{min}(p,q) = 2n-\vert k \vert$, this tensor network represents the periodic ground state $\ket{\Psi_{k}}$. 
	
	One way to verify this claim would be to compute all valid tilings of this network. The triangle tensors can be redefined in tile notation as follows  
	\begin{align}
		T = \sum_{l=1}^{7} D_{l} \label{eq:tdef}
	\end{align}
	where the triangular tiles $D_{l}$ are defined in Fig.~{\ref{fig:rntndefs}(b)}, where 
	\begin{align}
		\delta_{\vec{k}}(x, y, z) = \delta_{k_{1}, x} \delta_{k_{2}, y}\delta_{k_{3}, z}. \label{eq:delta2}
	\end{align}
	We use the zipper lemma to show in detail how the height renormalization tensor network is equivalent to the binary-height tensor network representation of the periodic model in Fig.~\ref{fig:TNequivalence}.
	
	\begin{figure}
		\begin{center}
			\includegraphics[width=\linewidth]{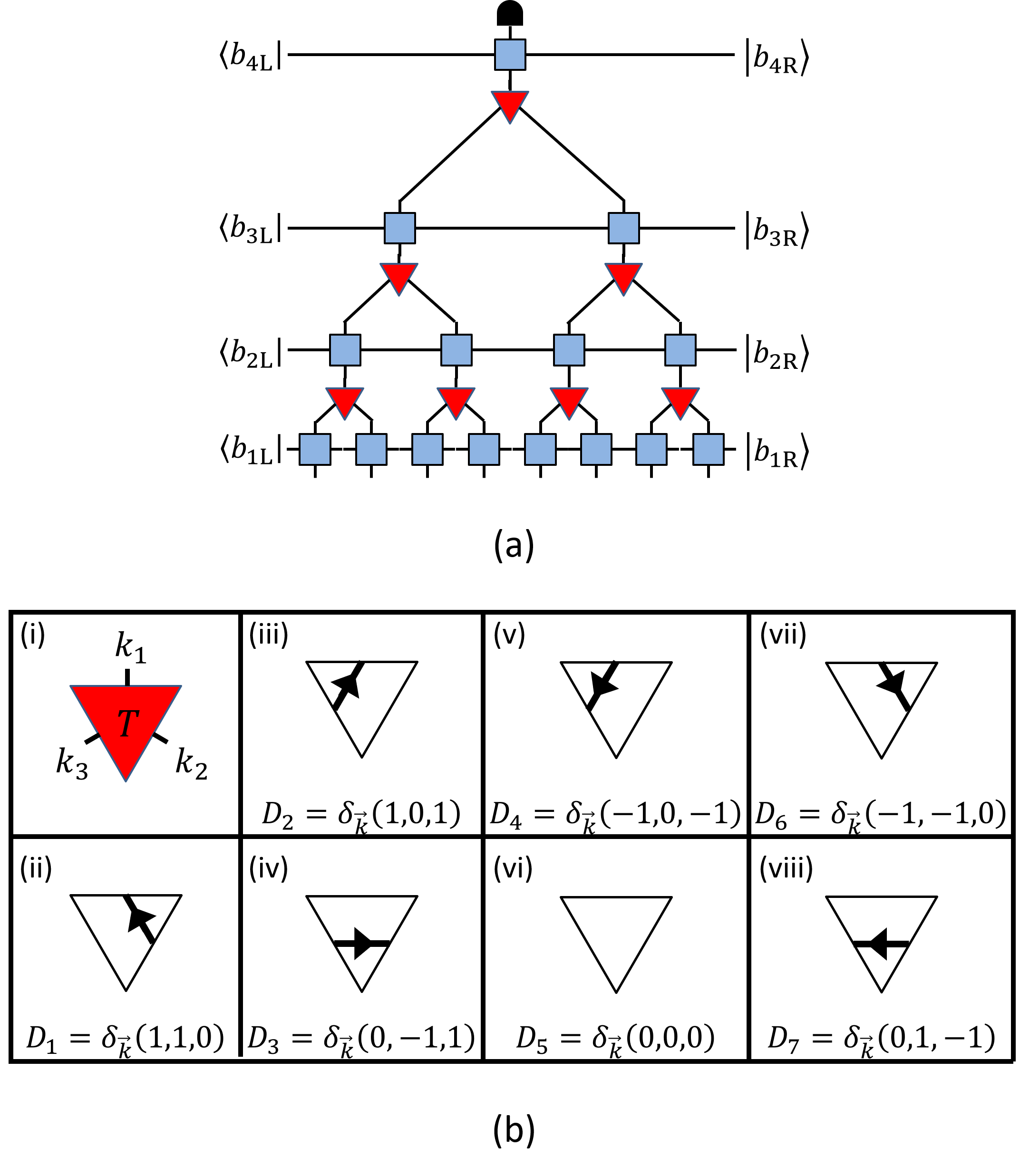}
		\end{center}
		\caption{$\mathbf{(a)}$ The height renormalization tensor network. The square tensors $B$ are identical to those defined in Fig.~{(\ref{fig:TN2defs})(b)}. $\mathbf{(b)}$ (i) The triangular tensors $T$ are defined in Eq.~(\ref{eq:tdef}). This definition involves the triangular tiles (ii-viii), which are themselves defined using $\delta_{\vec{k}}(x, y, z)$ (see Eq.~(\ref{eq:delta2})).} \label{fig:rntndefs}
	\end{figure}
	
	Next, we show that by modifying the $B$ and $T$ tensors, the same type of tensor network can represent the original Motzkin model (with open boundary conditions).

	\subsection{Original Motzkin model}
	The main modification we make to the network is that we increase the bond dimension of all indices from three to four. In addition to the three spin values $\{-1, 0, 1\}$, each index can also be assigned a non-physical value labeled $\omega$. Each tensor will still be represented as the sum of tiles. New tiles will indicate the $\omega$-value of an edge using a dotted line. In order to still represent a spin-1 chain, physical indices of the network (those that appear at the bottom) must take only spin-1 values. To ensure this, we define the projector
	\begin{align}
		\Pi = 1-\ket{\omega}\!\!\bra{\omega}, \label{eq:Piproj}
	\end{align}
	which will be appended to all physical indices of the network. 
	
	We define the new tensor network in Fig.~{\ref{fig:rntndefs2}(a)}. We modify the square and triangular building blocks of the tensor as follows. 
	
	The square tensors in Fig.~{\ref{fig:rntndefs2}(a)} are denoted $C$ and are defined as 
	\begin{align}
		C= B +  A_{7} + A_{8} \label{eq:TN2Bdef}
	\end{align}
	where $A_{7}$ and $A_{8}$ are defined in Fig.~{\ref{fig:rntndefs2}(b)}, and $B$ was defined in Eq.~(\ref{eq:TNBdef}). 
	The triangle tensors in Fig.~{\ref{fig:rntndefs2}(a)} are denoted $S$, and are defined as  
	\begin{align}
		S= \sum^{12}_{l=1} E_{l} \label{eq:Sdef}
	\end{align}
	where $E_{1},\dots E_{12}$ are defined in Fig.~{\ref{fig:rntndefs2}(c)}.

	\begin{figure*}
		\begin{center}
			\includegraphics[width=\linewidth]{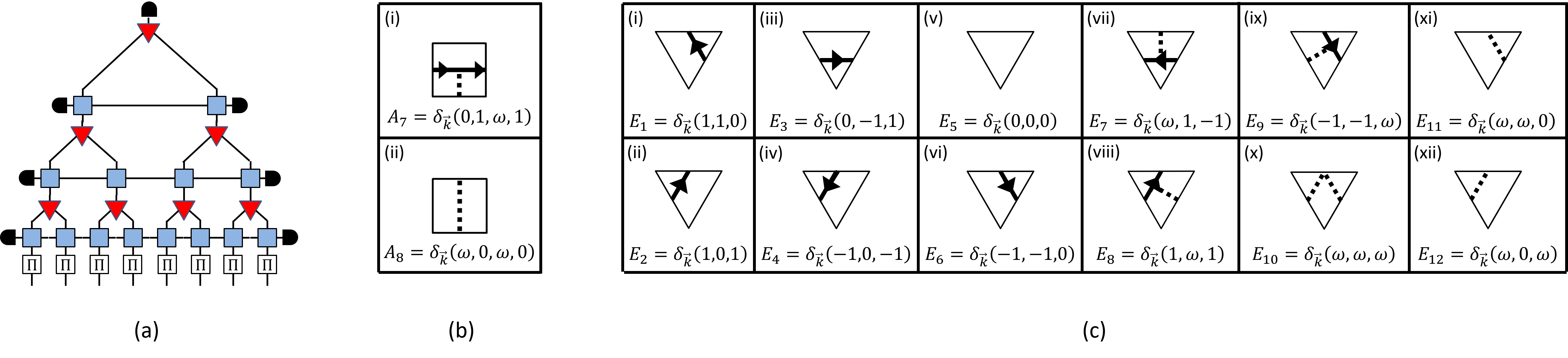}
		\end{center}
		\caption{$\mathbf{(a)}$ Tensor network representation of the spin-1 Motzkin model. The definitions of the $\Pi$ tensors, blue square $C$ tensors, and red triangular $S$ tensors are given in Eqs.~(\ref{eq:Piproj}), (\ref{eq:TN2Bdef}), and (\ref{eq:Sdef}). $\mathbf{(b)}$ Two square tiles used in the $S$ tensor. The value of $\omega$ is represented by a dotted line.  $\mathbf{(c)}$  Twelve triangular tiles used in $S$. }\label{fig:rntndefs2}
	\end{figure*}
	
	Similar to the height renormalization tensor network for the periodic case, we verify in Fig.~\ref{fig:TNequivalence2} that this network represents the Motzkin ground state by proving its equivalence to the binary-height tensor network shown in Fig.~\ref{fig:TN2defs}. 
	
	We also note that the height renormalization tensor network can be generalized to the spin-1/2 Fredkin model in the same way as the binary-height tensor network. The generalization to area-weighted walks is given in Fig.~\ref{eq:rnweight}. 
	
	\section{Discussion}
	
	We presented exact hierarchical tensor networks for representing the ground state of the spin-1 Motzkin spin-chain. These networks generate the sum over tile configurations that are one-to-one with valid Motzkin walks. 
	The networks utilize the characterization of the states as height models, the ability to encode the height efficiently in a binary way, and that binary addition can be understood as a local operation between digits in an addition algorithm. 
	The tile-based two-dimensional representation of each walk provides a bulk description of the spin chain: each valid bulk ``picture'' corresponds to a particular boundary state in the ground state superposition. It is interesting to note that bulk boundary correspondence in terms of tiling has played a useful theme in tackling other hard problems, most recently in the context of spin glasses \cite{klich2014glassiness}. 
	
	At this point it is interesting to consider the relation between our scale-invariant network, the renormalization group (RG), and the MERA. The height renormalization network we describe is clearly similar in structure to a 1D binary (scale-invariant) MERA~\cite{evenbly2009algorithms}, but where the disentangling has been realized through the action of a matrix product operator (MPO) as opposed to the tensor product of local unitary gates utilized in a standard MERA. Indeed, the alternative construction of a network for the Motzkin chain described in Figs.~\ref{fig:TensorDef}-\ref{fig:Compare} proceeds from the perspective of an RG transformation, building the network layer-by-layer in a manner comparable to that constructions of MERA using entanglement renormalization~\cite{Vidal2007ER}, thus further illuminating the connection between these ideas.  However, a significant difference is that the tensors in our network lack the isometric constraints imposed on MERA tensors, which are responsible for the finite-width ``light cone'' in MERA. Therefore our network does not preserve locality when viewed as an RG transformation, and should be viewed as a non-unitary MERA.
	
	The results presented here offer an exact analytic hierarchical tensor network representation of the ground state of
	a gapless system in the thermodynamic limit. Related works include networks based on network realizations of the Fourier transform~\cite{ferris2014fourier} and recent works on holographic tensor networks for disordered systems~\cite{jahn2020central}. While recent advances utilizing wavelets have allowed for the analytic construction of MERA that approximate the ground states of certain lattice CFTs with arbitrarily high precision ~\cite{Evenbly2016,haegeman2018rigorous}, these constructions only become exact in the limit of infinite bond dimension. The inability of a finite bond dimension MERA to achieve an exact representation of a CFT can be understood as a direct consequence of the finite-width light cone in MERA, which, at finite bond dimension, can only support a finite number of scaling operators~\cite{pfeifer2009entanglement}. This is incompatible with achieving an exact representation of a CFT, which typically possess an infinite set of scaling operators. However, this suggests that one should consider hierarchical networks of the form derived in this Article, which replace the (local) unitary disentangling with a (non-local) disentangling implemented by an MPO, more generally, as these do not have the limitation of only supporting a finite number of scaling operators. Indeed, our present work opens the exciting possibility that other systems, potentially including ground states of lattice CFTs, could also have an exact representation as a (finite bond dimension) network of this form. Notice also, that if the MPO disentangler was required to be a unitary operator~\cite{cirac2017matrix,csahinouglu2017matrix}, one could achieve a quasi-local RG transformation that may still be computationally viable as a variational ansatz. It thus remains an interesting avenue for future research to investigate whether a non-unitary hierarchical tensor network ansatz, which sacrifices the exact locality present in a standard MERA, could lead to improved numeric simulation algorithms.

	It has been noted that scale-invariant networks, and in particular MERA, have a special connection to holographic duals in the sense of the AdS/CFT correspondence \cite{swingle2012entanglement}. Here, the bulk of a MERA tensor network can be understood as a discrete realization of 3d anti-de Sitter space (AdS$_{3}$), identifying the extra radial holographic dimension with the RG group scale dimension in the MERA. %
	While the MERA tensor network is used to represent ground states of relativistic CFTs, where both time and space may be rescaled simultaneously, the tensor networks presented in this work perhaps more naturally describe ground states of non-relativistic field theories that are invariant under non-relativistic symmetry groups. These non-relativistic symmetry groups are characterized by anisotropic scale transformation of time and space 
	Lifshitz and Schrodinger field theories are well-known examples of such field theories.
	This relation may come to be valuable to the currently active attempts of constructing holographic duals for non-relativistic field theories with Lifshitz symmetry \cite{taylor2008non}
	or Schrodinger symmetry \cite{balasubramanian2008gravity,son2008toward}. 
	For a recent thorough review of Lifshitz holography, see \cite{taylor2016lifshitz} and references therein.
	
	It is natural to ask whether the tensor networks presented here can be generalized to other spin models. In another work~\cite{alexander2018rainbow}, we introduce another tile-based tensor network for the higher spin (colored) generalizations of the Motzkin and Fredkin models, showing how an exact network can describe the rainbow phases these models exhibit. %
	
	Our results highlight the versatility of tensor networks and their potential in describing complicated many body states, and in particular the power of self-similar structures such as MERA, or, in this case, a non-unitary hierarchical tensor network, in describing quantum critical phases.

	\emph{Acknowledgments.}
	We thank Amr Ahmadain, Zhao Zhang, Jacob Bridgeman, Elizabeth Crosson, Joseph Avron, Ramis Movassagh, Vladimir Korepin, Diana Vaman and Xiao-Liang Qi for discussions.
	The work of I.K. was supported in part by the NSF grant DMR-1508245 and NSF grant DMR-1918207.
	The work of R.N.A. was partially supported by National Science Foundation Grant No.~PHY-1630114. 
	The work of G.E. was undertaken thanks in part to funding from the Canada First Research Excellence Fund (CFREF).

	\newpage
	\appendix
	\onecolumn
	

	%
	\section{Definition of the Motzkin Hamiltonian}\label{app:A}
	The Motzkin Hamiltonian can be defined by first identifying each local spin-$z$ basis state $\{ \ket{1}, \ket{0}, \ket{-1}\}$ with a line segment $\{ \diagup,  \text{\textbf{---}}, \diagdown \}$, respectively. This allows us to represent states as a superposition of walks. The Hamiltonian, defined on a spin-chain with $2n$ sites reads:
	\begin{equation} H= \Pi_{\text{boundary}} + \sum_{j=1}^{2n-1}\Pi_{j} \label{cHam} 
	\end{equation} 
	where  $\Pi_{j}$ acts on the pair of spins $j,j+1$ and
	\begin{align*} 
		\Pi_{j} &= 
		|\Phi_{t}\rangle\langle\Phi_{t}|_{j,j+1} + |\Psi_{t}\rangle\langle\Psi_{t}|_{j,j+1} + |\Theta_{t}\rangle\langle\Theta_{t}|_{j,j+1},\\
		\Pi_{\text{boundary}} &= |-1\rangle\langle -1 |_1 + |1\rangle\langle 1|_{2n}.
	\end{align*} where  $\Phi,\Psi,\Theta$ are the following states on pairs of neighboring spins 
	\begin{align} 
		|\Phi_{t}\rangle &\propto |1,0\rangle - t|0,1\rangle, \\ |\Psi_{t}\rangle &\propto |0,-1\rangle - t|-1,0\rangle, \\ |\Theta_{t}\rangle &\propto |1,-1\rangle - t|0,0\rangle. \label{statess}
	\end{align} %
	
	To have periodic boundary conditions, the Hamiltonian in Eq.~(\ref{cHam}) can be simply modified to include a $\Pi_{2n}$ term while omitting the boundary terms $\Pi_{\text{boundary}}$.

	\section{Spin 1/2 case}\label{sec:TN2fred}
	To re-purpose the binary-height network for the spin-1/2 Fredkin model, we define  an operator $P\in \text{Hom}(\mathbb{C}^{3}, \mathbb{C}^{2})$ that projects out the $0$ component of each spin and maps $1\mapsto \tfrac{1}{2}$ and $-1 \mapsto -\tfrac{1}{2}$. It is given by 
	\begin{align}
		\begin{tikzpicture}
			\begin{scope} 
				\node {$\displaystyle P = \ket{\tfrac{1}{2}}\bra{1} + 
					\ket{-\tfrac{1}{2}}\bra{-1}=$};
			\end{scope}
			\begin{scope}[xshift = 2.4cm]
				\node {\includegraphics[width=0.03\linewidth]{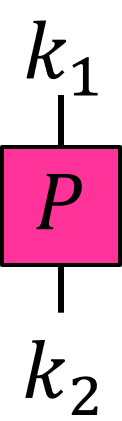}};
			\end{scope}
		\end{tikzpicture} \label{eq:Pop}
	\end{align}
	where the indices $k_{1}$ and $k_{2}$ are associated with the spin-1 and spin-$\tfrac{1}{2}$ degrees of freedom, respectively. Appending $P^{\otimes 2n}$ to the bottom of the binary-height tensor network yields a tensor network for the Fredkin model: 
	\begin{align}
		\includegraphics[width=0.35\linewidth]{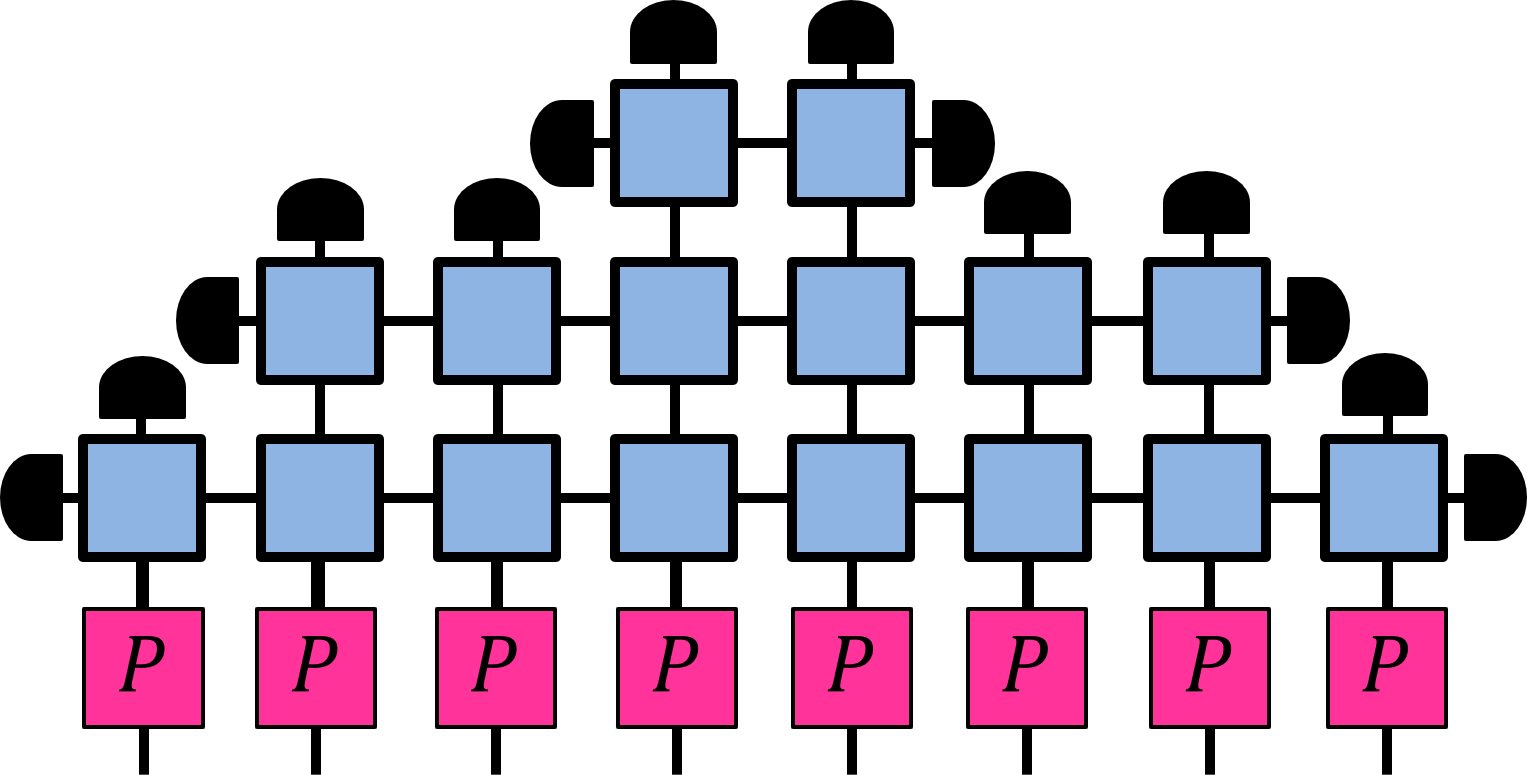}\label{eq:fredbintens}
	\end{align}
	
	We note that Salberger and Korepin previously presented an MPS description of the Fredkin state using the height representation of Dyck walks~\cite{salberger2017entangled}. However, the elementary building blocks of such networks have growing bond-dimension and lack the local structure of the networks we achieve here. 
	
	\section{Construction of renormalization tensor network via U(1)-symmetric tensors}\label{sec:const2}

	In this appendix we provide an alternative derivation of an exact hierarchical tensor network representation of the ground state of the spin-$1$ Motzkin spin chain. Here, in order to better connect with established tensor network methodology~\cite{singh2011tensor}, we formulate the solution in terms of U(1) invariant tensors instead of the `flux' preserving tiles discussed in the main text; however both approaches are ultimately equivalent. The derivation presented here is based on the dual notions of $\delta$-symmetric tensors, a restricted subclass of U(1) tensors, and {\it boundary-locked} networks. We first define these concepts before proceeding to demonstrate how they can be used to construct a network for the Motzkin spin chain.
	
	Let us recall the basics of U(1)-symmetric networks, as detailed in Ref.\cite{singh2011tensor}. A tensor network with U(1) symmetry is represented by an {\it oriented graph}, where each index has an associated direction (depicted with an arrow), such that the indices connected to a tensor can be regarded as either incoming or outgoing with respect to that tensor. Each index $i \in \{ 1,2,\ldots,d\}$ in the network is assigned a set of quantum numbers $\vec n^{(i)} = [n^{(i)}_1, n^{(i)}_2, \dots, n^{(i)}_d]$ with charges $n \in \mathbb{Z}$ associated to the $z$-component of the spin at that value of the index. We say that a tensor is U(1) symmetric if it is left invariant under transformation with a unitary representation of U(1) acting on each index (where outgoing indices should be transformed by the dual representation of U(1) with respect to incoming indices). It is known that U(1)-symmetric tensors are those that conserve particle number, such that a tensor component has zero weight unless the sum of the outgoing charges matches the sum of the incoming charges. For the purposes of constructing a solution to the Motzkin model, it is useful to restrict to a sub-class of symmetric tensors that we call $\delta$-{\it symmetric} tensors, which we define as U(1)-symmetric tensors where every (structurally non-zero) element is equal to unity. In other words, these are U(1)-symmetric tensors with component equal to unity if the total incoming charge $n_\textrm{in}$ matches the total outgoing charge $n_\textrm{out}$ and zero otherwise, which can be understood as a tensor version of the Kronecker-delta function.
	
	The concept of a {\it boundary-locked} tensor network is now defined. Given a lattice $\mathcal L$ of spin-1 sites, let $T$ be the tensor network built from U(1)-symmetric tensors, representing a $S_z = 0$ quantum state $\ket{\psi} \in \mathcal L$, and let $\ket{\phi}$ be a U(1)-invariant product state ($S_z = 0$) on $\mathcal L$. We say that network $T$ if {\it boundary-locked} if, by constraint of the U(1) symmetry, there is a single unique configuration of the internal indices in $T$ that can give a non-zero contribution to the scalar product $\braket{\phi \vert \psi}$ for any input product state $\ket{\phi}$ with $S_z = 0$. In other words, a network is boundary locked if, by fixing the boundary indices in a configuration compatible the $S_z = 0$ symmetry of the network, the internal indices are then `locked' in a unique configuration by the constraint that the total incoming U(1) charges must match the total outgoing charge in all tensors. It follows that a necessary (but not sufficient) constraint for a network to be boundary locked is that the irrep of charge $n$ on any tensor index is at most 1-fold degenerate (which allows us to specify the value that an index takes by the charge $n$ that it carries). 
	
	\begin{figure}[!t]
		\begin{center}
			\includegraphics[width=8cm]{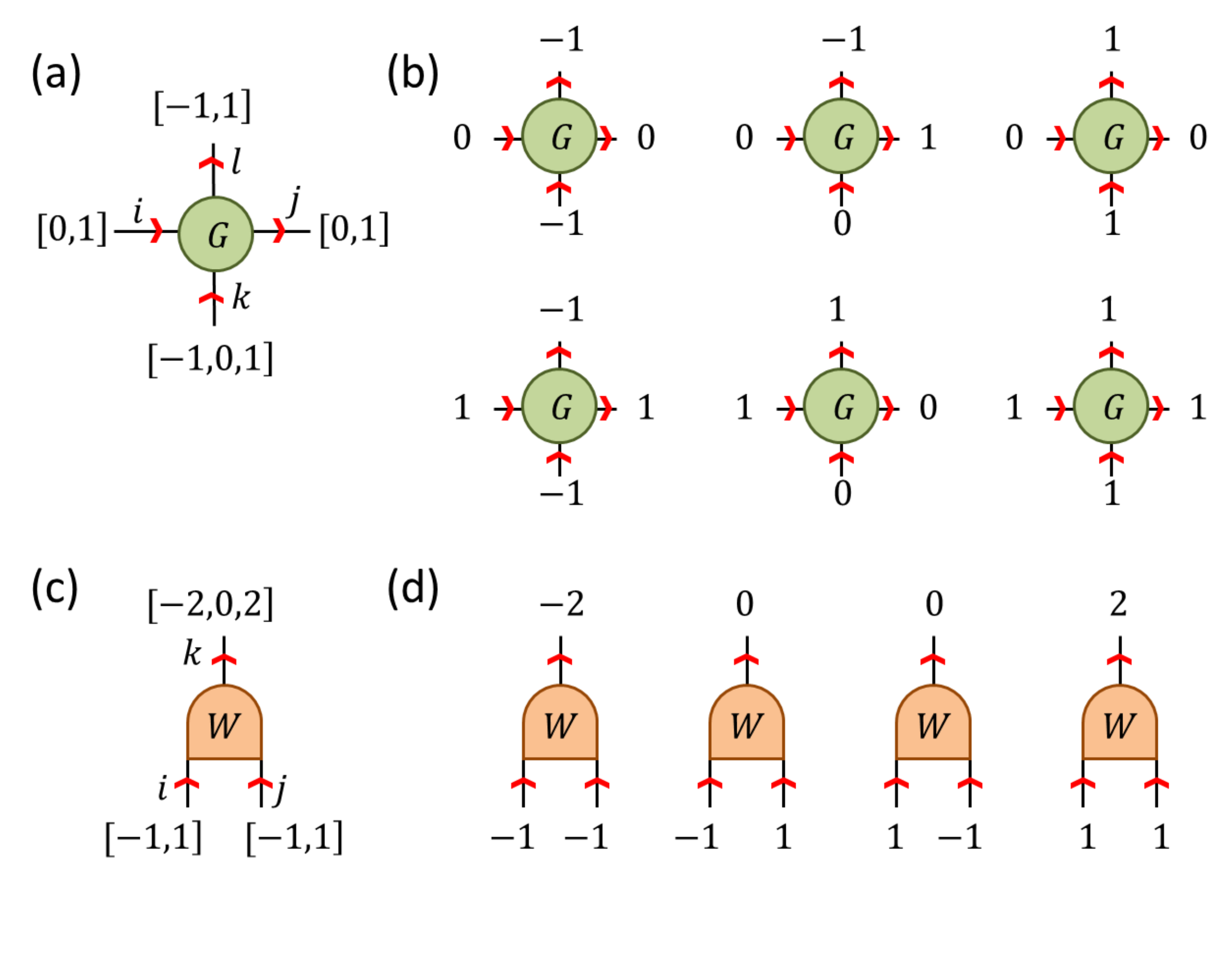}
			\caption{(a) Description of the U(1) charges on each index of the $G^{ij}_{kl}$ tensors. (b) List of the non-zero components of $G$, which have equal incoming and outgoing U(1) charge. (c) Description of the U(1) charges on each index of the $W^{ij}_k$ tensors. (b) List of the non-zero components of $W$. These have equal incoming and outgoing U(1) charge. }
			\label{fig:TensorDef}
		\end{center}
	\end{figure}
	
	Given that the notions of $\delta$-symmetric tensors and boundary-locked networks have been established, we are now able to construct the exact hierarchical network for the ground state of the Motzkin chain. We begin by focusing on a simpler task: given a finite chain of spin-$1$ sites $\mathcal L$ we describe how to construct a network representing the equal weight superposition U(1) states in the $S_z = 0$ spin symmetry sector (or equivalently, the superposition of all walks
	that start and end at zero height). This solution can later be refined to exclude the paths that take negative height values at any point, as described in the main text, such that the ground state of the original Motzkin spin chain is recovered. Let us consider a U(1)-invariant state $\ket{\psi} \in \mathcal{L}$ with $S_z = 0$; if $\ket{\psi}$ is described by a boundary locked network built from $\delta$-symmetric tensors it automatically follows that $\ket{\psi}$ must be the desired equal-weight superposition of all states in the $S_z = 0$ sector. This is true since the scalar product of $\ket{\psi}$ with any product state of the $S_z = 0$ spin sector must evaluate to unity, as only a single configuration of indices can contribute and all configurations have a total weight equal to unity, while the scalar product with any state outside of the $S_z = 0$ is trivially zero.
	
	\begin{figure}[!t]
		\begin{center}
			\includegraphics[width=8cm]{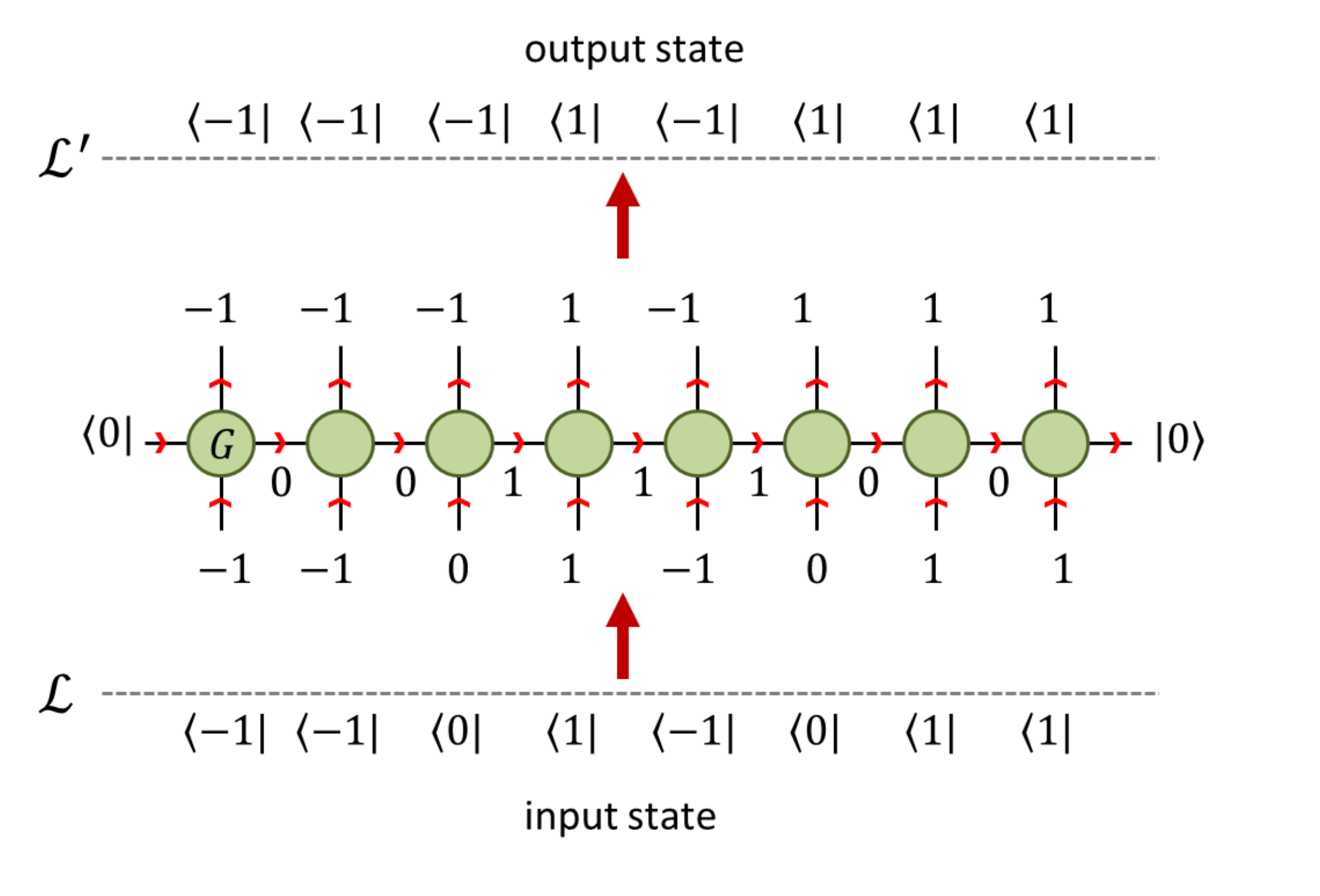}
			\caption{The MPO formed from $G$ tensors maps an input product state on lattice $\mathcal L$ of local dim $d=3$ to a single product state on lattice $\mathcal L'$ of local dim $d=2$.}
			\label{fig:MPO}
		\end{center}
	\end{figure}

	The remaining goal is thus is to build a boundary-locked network $T $of $\delta$-invariant tensors that constitutes a proper holographic realization: given a lattice $\mathcal L$ of $N$ sites, we want the network $T$ to be organized into $O(\log N)$ self-similar layers, each with some finite bond dimension that is independent of the system size $N$. Here we follow a similar construction as used in the MERA, and build the network from a sequence of coarse graining (CG) transformations, each of which is comprised of a {\it disentangling} step followed by a blocking step that reduces the number of lattice sites by a factor of $2$. However, instead of using disentanglers with finite local support, which prove inadequate for the exact construction, we instead represent the disentangling operation as an MPO. We build this MPO from copies of a four index $\delta$-invariant tensor $G^{ik}_{jl}$ as depicted in Fig.~{\ref{fig:TensorDef}(a-b)}. The virtual bond dimension of the MPO is set at $d=2$, with U(1) charges on these indices $\vec n^{(i)} = \vec n^{(j)} = [0,1]$. The incoming indices $k$ must match the spin-$1$ sites in $\mathcal L$, i.e. such that $\vec n^{(k)} = [-1,0,1]$, while the output indices $l$ are set at $d=2$ with U(1) charges $\vec n^{(l)} = [-1,1]$. Thus the action of this MPO is to map states on the $N$-site lattice $\mathcal L$, which has local dimension $d=3$, to states on an $N$-site lattice $\mathcal L'$ of local dimension $d=2$. If we assume $N$ to be even and that the virtual indices of the edge MPO tensors are fixed in the $\ket{0}$ state, then it can be seen that any $S_z = 0$ product state on $\mathcal L$ is mapped to a unique product state on $\mathcal L'$, see Fig.~\ref{fig:MPO} for example, consistent with the desire for a boundary-locked network. This follows as, once the input indices $i$ and $k$ have been set on a $G$ tensor, there is only a single choice of the output indices $j$ and $l$ that satisfies U(1) charge conservation. The blocking step of the CG transformation is now realized using a 3-index $\delta$-symmetric tensor $W^{ij}_k$, whose input indices carry U(1) charges $\vec n^{(i)} = \vec n^{(j)} = [-1,1]$ and whose output index carries $\vec n^{(k)} = [-2,0,2]$, see also Fig.~{\ref{fig:TensorDef}(c-d)}. This blocking step is again consistent with realizing a boundary locked network, and clearly has non-zero overlap with any state on $\mathcal L'$.
	
	\begin{figure}[!t]
		\begin{center}
			\includegraphics[width=8cm]{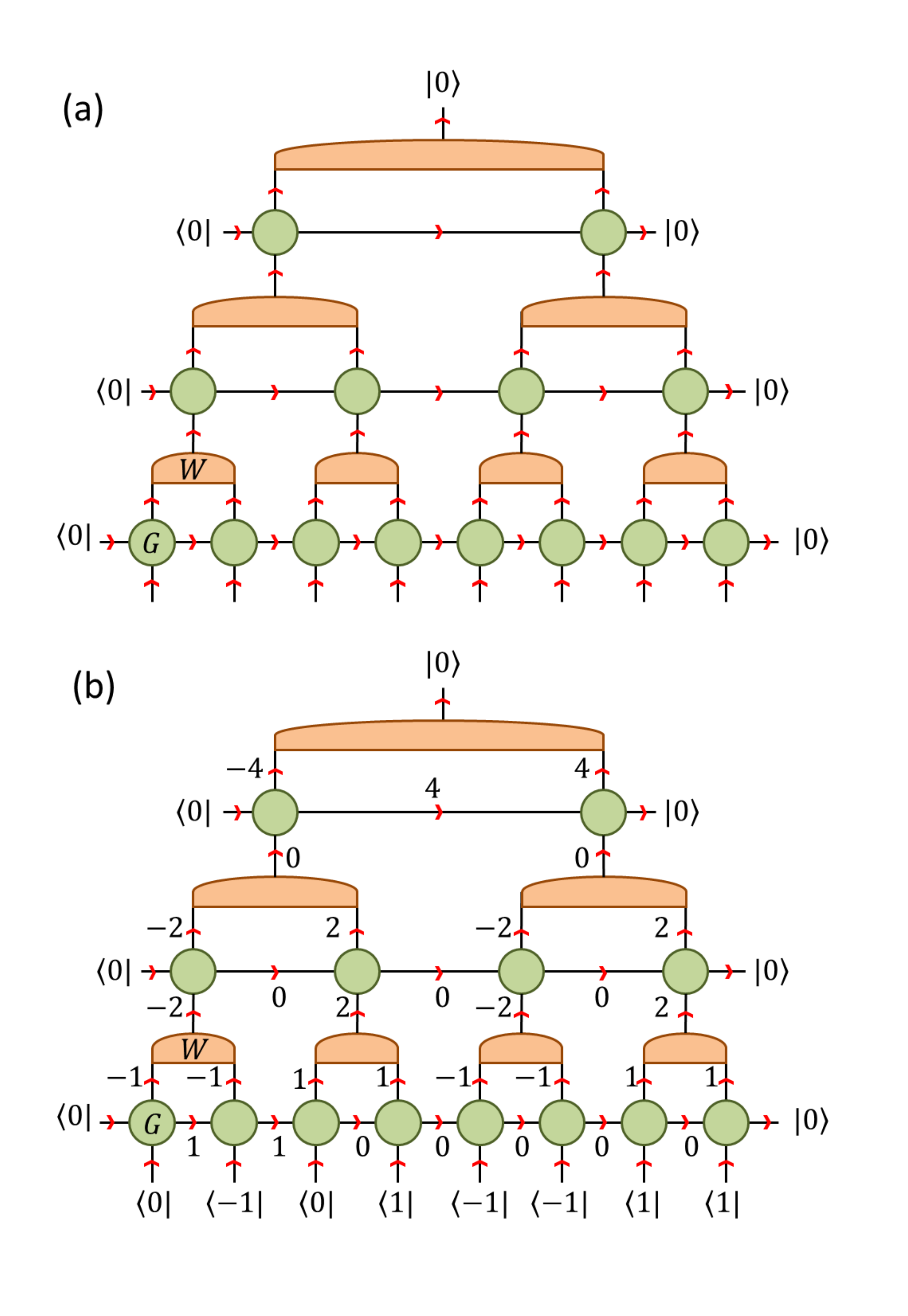}
			\caption{(a) Depiction of the holographic tensor network. (b) The network is {\it boundary-locked}; under the input product state (in the total spin $S_z = 0$ sector) the internal indices are fixed in a single unique configuration in order to satisfy the U(1) symmetry constraints on each tensor. }
			\label{fig:Network}
		\end{center}
	\end{figure}
	
	\begin{figure}[!t]
		\begin{center}
			\includegraphics[width=6cm]{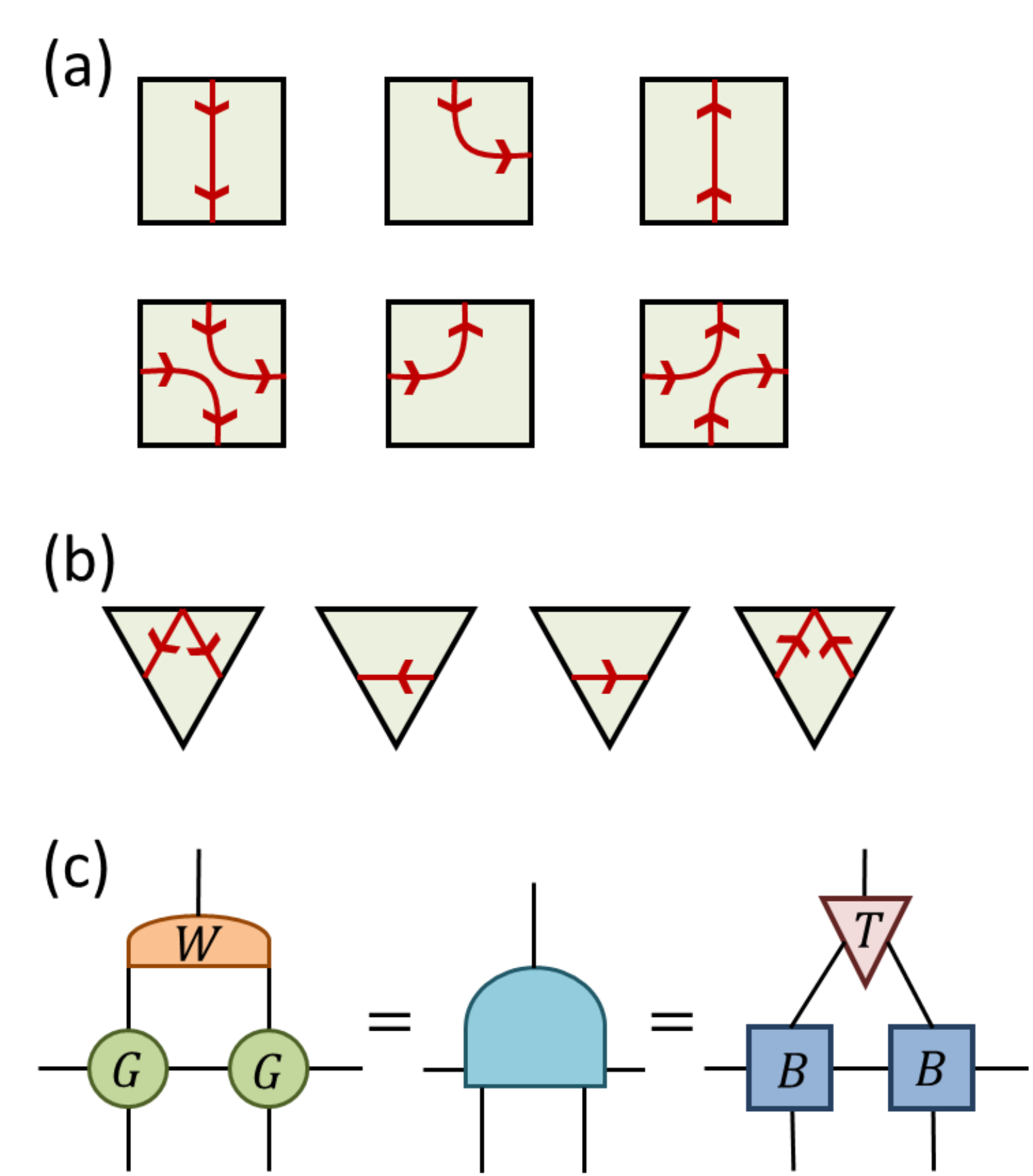}
			\caption{(a) Tile representation of the components of the $G$ tensor from Fig.~{\ref{fig:TensorDef}(b)}. (b) Tile representation of the components of the $W$ tensor from Fig.~{\ref{fig:TensorDef}(d)}. (c) The product of $G$ and $W$ tensors is identical to the product of $B$ and $T$ tensors from Eq.~(\ref{eq:renormunit}).}
			\label{fig:Compare}
		\end{center}
	\end{figure}
	
	One may then repeat this coarse graining transformation, consisting of disentangling with the MPO formed from $G$ tensors and then blocking with the $W$. Notice that the magnitude of U(1) individual charges carried on the indices of the $G$ and $W$ tensors doubles with each coarse graining step, but the content of these tensors otherwise remains unchanged. The output index on the final isometry, after taking $\log_2(N)$ coarse graining steps where $N$ is assumed to be a power of 2, is fixed in the $n = 0$ state to ensure that the network is in the total spin $S_z = 0$ sector, see Fig.~\ref{fig:Network}. Thus our goal is completed: we have an $S_z = 0$ holographic network that (i) is constructed entirely from $\delta$-symmetric tensors and (ii) is boundary-locked, which implies that it represents an equal weighted superposition of all states in the total spin $S_z = 0$ sector.
	
	In comparing the network derived in this appendix, Fig.~{\ref{fig:Network}(a)}, to that derived in the main text, Fig.~\ref{fig:rntndefs}, one sees that they both have an equivalent structure, with the $G$ and $W$ tensors derived in this appendix substituting for the $B$ and $T$ tensors in Fig.~\ref{fig:rntndefs}. However, upon expressing the $G$ and $W$ tensors in the `tiling' representation used in the main text, see Fig.~{\ref{fig:Compare}(a-b)}, it is seen that they correspond to a different set of tiles than do the $B$ tensors, Fig.~\ref{fig:TN2defs}, and the $T$ tensors, Fig.~\ref{fig:rntndefs}. Nevertheless, a component-wise analysis of all permissible tilings reveals that the product of two $G$ and a $W$ tensor is identical to the product of two $B$ and a $T$ tensor, as depicted in Fig.~{\ref{fig:Compare}(c)}. Thus one indeed concludes that the network constructed in this appendix is ultimately equivalent to that of the main text.

	\section{Equivalence between the renormalization and binary height networks for the periodic model}\label{app:periodicequiv}
	Here we prove that the height renormalization tensor network defined in Fig.~\ref{fig:rntndefs} represents the periodic Motzkin ground state. We do this by showing its equivalence to the generalized binary-height tensor network defined in Fig.~{\ref{fig:TN2defs}(b)}.

	In order to represent the state $\ket{\Psi_{k}}$, the boundary vectors $\ket{\vec{b}_{\text{L}}}$ and $\ket{\vec{b}_{\text{R}}}$ (see Eq.~(\ref{eq:bdvec})) are chosen to encode integers $p$ and $q$ such that $q-p=k$ and $\text{min}(p, q)\geq n$. We choose
	\begin{align}
		\text{max}(p,q) &= 2^{\lfloor \log_{2} 2n \rfloor +1} + 2n \label{eq:maxh}\\
		\text{min}(p,q) &= 2^{\lfloor \log_{2} 2n \rfloor +1} +2n -\vert k\vert,  \label{eq:minh}
	\end{align}
	We use a rectangular binary-height tensor network with $m=2^{\lfloor \log_{2} 2n \rfloor +1}$ many layers. This choice guarantees that $b_{m\text{L}}=b_{m\text{R}}=1$. See Fig.~{\ref{fig:TNequivalence}(a)} for an example on 8 spins.
	
	As usual, the upwards pointing indices on the $m^{\text{th}}$ layer have been contracted with the state $\ket{0}$. In fact, all such indices {\it must} take the value 0 regardless of this boundary condition.  To see this, suppose some of these indices took a value of 1 or -1, and let $\mu$ be the leftmost column with a non-zero value at the top. 
	
	Then,
	\begin{itemize}
		\item If the top index of $\mu$ takes value $-1$, the only valid tiling for $\mu$ is to use tile $A_{6}$ for all square cells. Then, the height encoded at the left boundary of $\mu$ is zero. However, from Eq.~(\ref{eq:minh}), $p> 2n \geq \mu$, and since it takes at least $y$ many columns to drop in height by $y$, there are no valid tilings of the lattice to the left of $\mu$. 
		
		\item If the top index of $\mu$ takes value $1$, the only valid tiling for $\mu$ is to use tile $A_{2}$ for all square cells. Then, the height encoded at the right boundary of $\mu$ is zero. However, from Eq.~(\ref{eq:minh}), $q>2n\geq 2n - \mu$, and since it takes at least $y$ many columns to climb in height by $y$, there are no valid tilings of the lattice to the right of $\mu$.
	\end{itemize}
	Because non-zero values at the top indices of the network yield invalid tilings, contractions of the network must evaluate to zero in such cases. 
	
	Therefore, there is some flexibility in choosing the upper boundary condition, i.e., we are free to replace the contraction with $\ket{0}^{\otimes 2n}$ with any other tensor that has equal support over this state. In particular, we can use a binary tree of triangle $T$ tensors, as shown in Fig.~{\ref{fig:TNequivalence}(b)}. Tiling the entire pyramid with $D_{5}$ tiles shows that it has support over the $\ket{0}^{\otimes 2n}$ subspace.

	\begin{figure}
		\begin{center}
			\includegraphics[width=\linewidth]{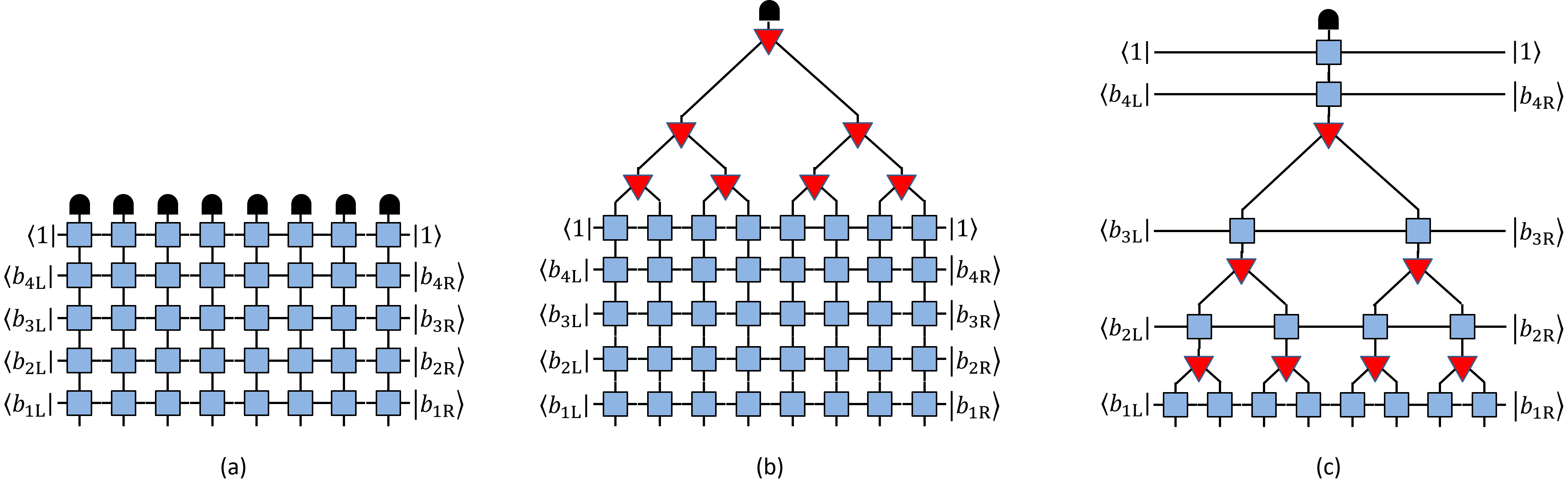}
		\end{center}
		\caption{$\mathbf{(a)}$ Rectangular binary-height tensor network with left and right boundary conditions corresponding to Eqs.~(\ref{eq:maxh}) and (\ref{eq:minh}). $\mathbf{(b)}$ Equivalent tensor network with the top boundary contracted with a binary tree of triangle $T$ tensors. $\mathbf{(c)}$ Tensor network after the zipper lemma has been applied to all $T$ tensors.} \label{fig:TNequivalence}
	\end{figure}

	Next, we require following lemma: 
	\begin{lemma}[Zipper Lemma]
		The following tensor networks are equivalent:
		\begin{align}
			\includegraphics[width=0.2 \linewidth]{zipperlemma2}\label{eq:zipperlemma}
		\end{align}
	\end{lemma}
	\begin{proof}
		By exhaustive search all valid tile configurations (see Fig.~\ref{eq:zipperlemmaproof}, the following two tensor networks are equivalent
		\begin{align}
			\includegraphics[width=0.2 \linewidth]{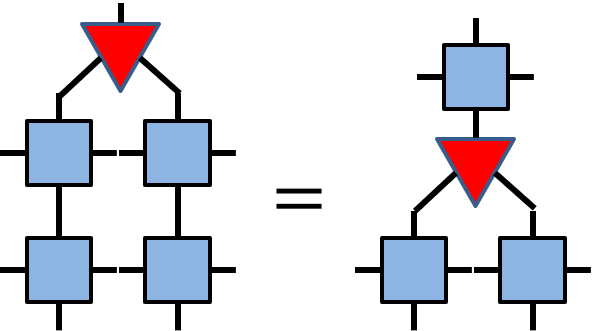}\label{eq:zippercorr}
		\end{align}
		The lemma follows by sequential application of this identity. 
	\end{proof}
	Using this, we can pull a horizontal layer of triangular tensors in Fig.~{\ref{fig:TNequivalence}(b)} downwards through the tensor network. At each step, two $B$ tensors are merged into one. Pulling down each triangular tensor is analogous to closing a zipper. The end result is shown in  Fig.~{\ref{fig:TNequivalence}(c)}.
	
	The final step to prove equivalence with the height renormalization tensor network from Fig.~{\ref{fig:rntndefs}(a)} is to note that the top $B$ tensor from Fig.~{\ref{fig:TNequivalence}(c)} can be removed, since
	\begin{align}
		\includegraphics[width=0.3 \linewidth]{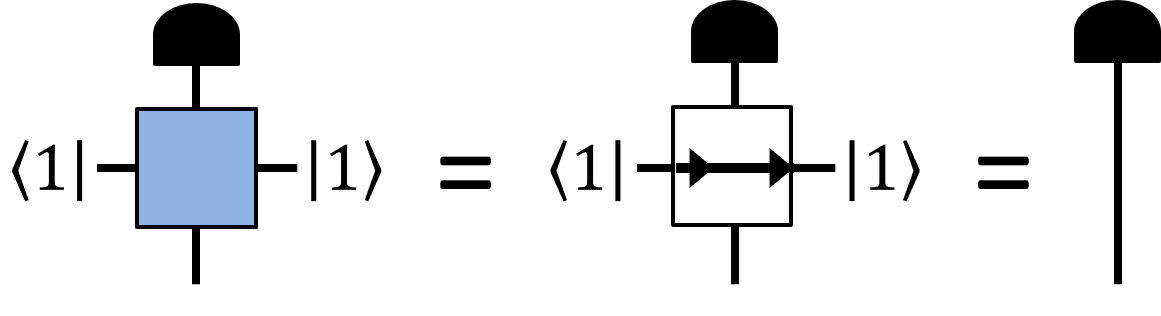}\label{eq:capremoval}
	\end{align}
	Removing the top $B$ tensor modifies the left and right boundary heights so that now $\text{max}(p, q) = 2n$ and $\text{min} (p, q) = 2n - \vert k\vert$.
	
	\begin{figure*}
		\includegraphics[width=1 \linewidth]{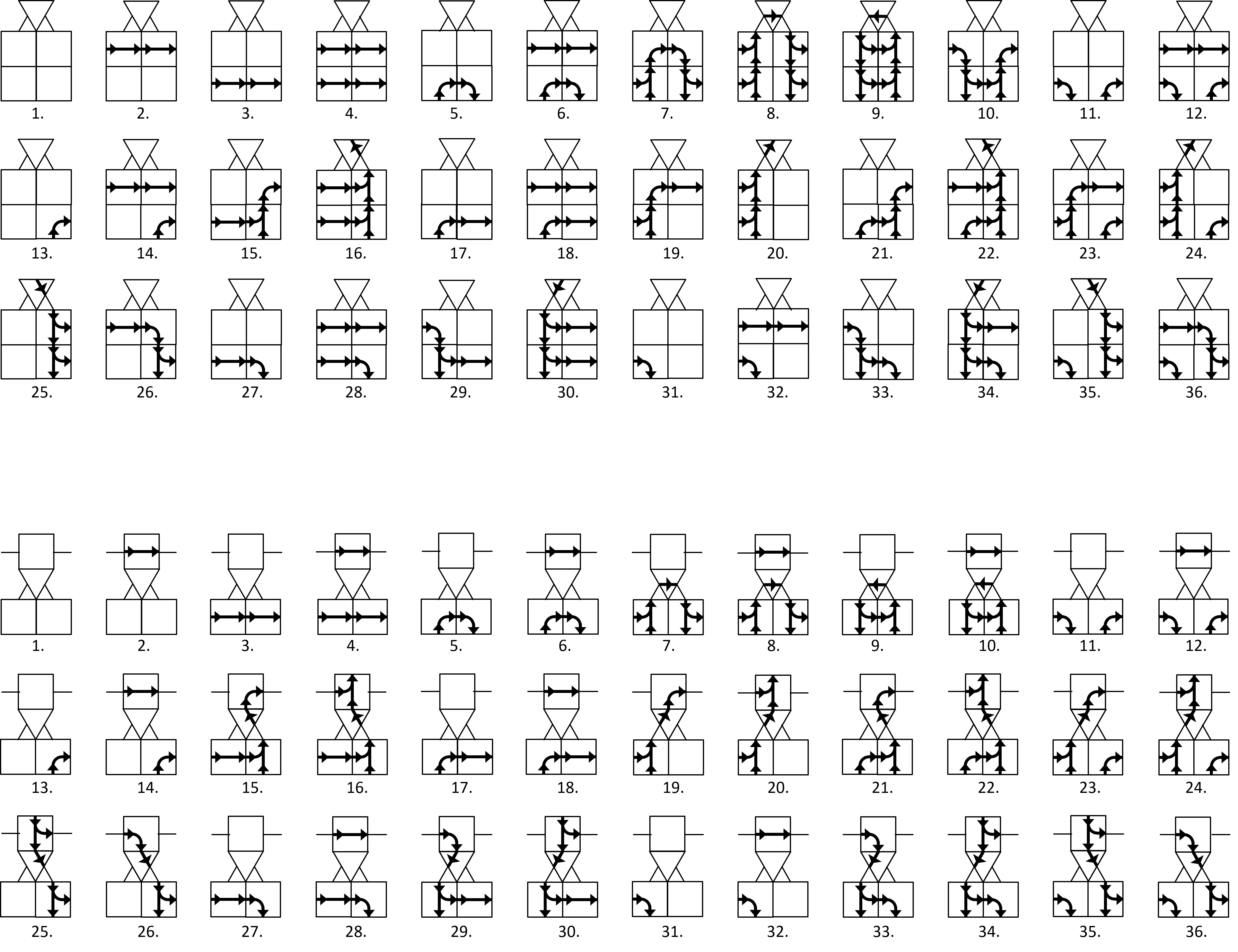}\caption{Each tiling of the networks in Eq.~(\ref{eq:zippercorr}) is uniquely specified by the values on the left and bottom pointing indices. These can take values $\{0, 1\}$ and $\{-1, 0, 1\}$, respectively. This results in a total of 36 possible tilings. We list these for both networks and find an isomorphism between configurations that share boundary conditions.}\label{eq:zipperlemmaproof}
	\end{figure*}
	
	\section{Equivalence between the renormalization and binary height networks for the original model}\label{sec:zipperproof2}
	Here we show that the height renormalization tensor network shown in Fig.~\ref{fig:rntndefs2} is a valid representation of the Motzkin model. Our proof will follow a similar trajectory to the case with periodic boundary conditions.
	
	We begin with the step-pyramid of $B$ tensors from Fig.~\ref{fig:TN2defs}. We will embed the width $2n$ pyramid within a $(\lfloor \log_{2} 2n \rfloor +1)\times 2n$ rectangle (recall that this does not change the represented state). We are also free to append projectors $\Pi$ (see Eq.~(\ref{eq:Piproj})) to the base of the network because the $B$ tensors contain no tiles with support on $\omega$ (they have no dotted lines). The tensor network shown in Fig~{(\ref{fig:TNequivalence2})(a)} incorporates both these changes.

	\begin{figure}
		\begin{center}
			\includegraphics[width=\linewidth]{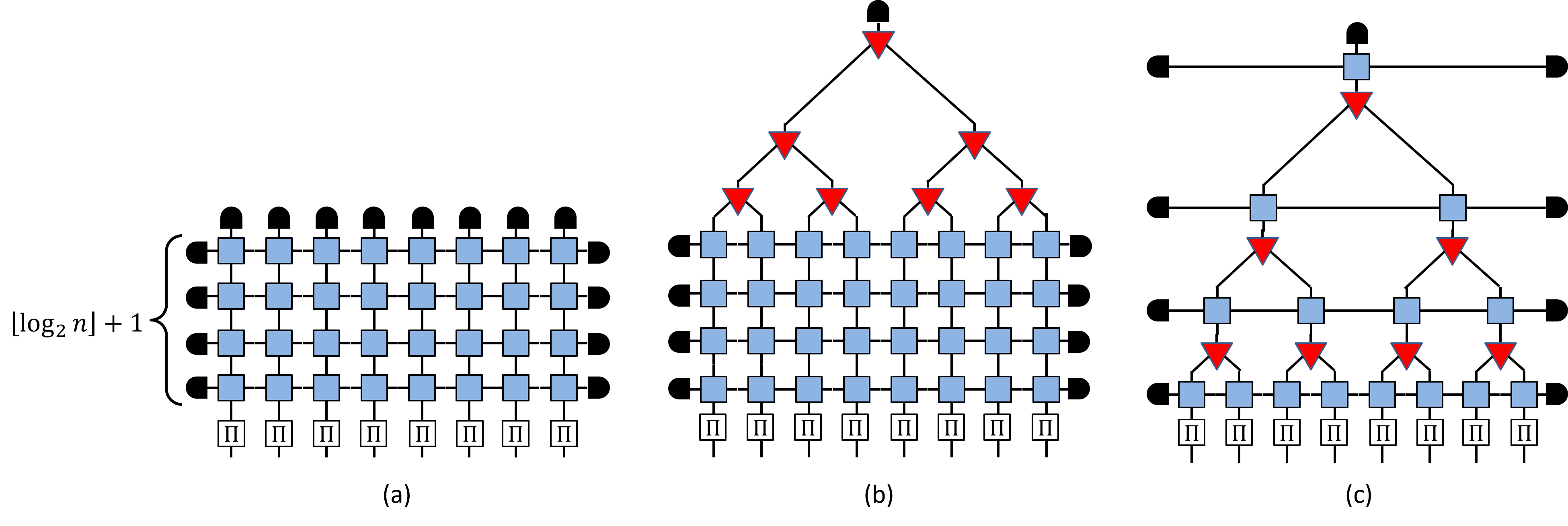}
		\end{center}
		\caption{$\mathbf{(a)}$ Binary height tensor network made up of $B$ tensors, vectors $\ket{0}$, and projectors $\Pi$. This is equivalent to the step-pyramid binary-height tensor network from Fig.~\ref{fig:TN2defs}. $\mathbf{(b)}$ The network from (a) has been modified so that the $B$ tensors are replaced with $C$ tensors, and the top boundary of $\ket{0}^{\otimes 2n}$ has been replaced with a tree of triangular $S$ tensors.}\label{fig:TNequivalence2}
	\end{figure}
	
	Next, we replace each of the $B$ tensors in Fig~{(\ref{fig:TNequivalence2})(a)} with $C$ tensors defined in Eq.~(\ref{eq:TN2Bdef}). Though the $C$ tensors contain two additional tiles ($A_7$ and $A_8$), the new network includes no additional tilings. To see this, note that $A_7$ and $A_8$ contain vertical dotted lines. If these tiles appear in some column, then any valid tiling of that column must connect the dotted line to the bottom of the network. Contraction with $\Pi$ at the bottom of that column will evaluate the network to zero. Next, we replace the top boundary $\ket{0}^{\otimes 2n}$ vector with a tree of triangular $S$ tensors (defined in Eq.~(\ref{eq:Sdef})). These changes are incorporated into the tensor network shown in Fig.~{\ref{fig:TNequivalence2}(b)}. This tensor network is equivalent to the network shown in (a). 
	
	To prove this, we will show that the indices at the bottom of the pyramid (equivalently, the indices at the top of the square lattice) in Fig.~{\ref{fig:TNequivalence2}(b)} must take the value zero in order for the network contraction to give a non-zero value. 
	Suppose some of these indices took a value of 1, -1, or $\omega$, and let $\mu$ be the leftmost column of the square lattice with a non-zero value at the top. 
	
	Then,
	\begin{itemize}
		\item If the top index of $\mu$ takes value $1$, the only valid tiling for $\mu$ is to use tile $A_{2}$ for all square cells. Then, the height encoded at the left boundary of $\mu$ is $2^{\lfloor \log_{2} 2n \rfloor +1}$. Since it takes at least $y$ many columns to climb in height by $y$, and the left boundary of the network is set to height zero, there is no valid tiling that is compatible with both the left boundary of the network and the right side of the column $\mu$. 
		\item If the top index of $\mu$ takes value $-1$, then the tiling of the pyramid network of $S$ tensors must include some $E_{7}$ tiles. Consider the topmost $E_{7}$ tile(s). Any valid tiling of the pyramid must connect this dotted line to the top of the pyramid. Then, contracting with the $\ket{0}$ vector at the top index of the top $S$ tensor will evaluate the network to zero.
		\item  If the top index of $\mu$ takes value $\omega$, then $\mu$ can be tiled only with $A_{8}$. Contraction with $\Pi$ at the base of $\mu$ evaluates the network to zero. 
	\end{itemize}
	Therefore, using a tree of $S$ tensors as the upper boundary of the square lattice network is equivalent to contracting with $\ket{0}^{\otimes 2n}$.
	
	Now, we reprove the zipper lemma (lemma~(\ref{eq:zipperlemma})) with tensors $C$ and $S$ replacing tensors $B$ and $T$, respectively. 
	
	\begin{proof}
		First we prove the equivalence of the tensor networks shown in Eq.~(\ref{eq:zippercorr}). We do this by proving an equivalence of tilings. The tilings of each network are specified by assigning values of $\{0, 1\}$ to the left facing indices and values of $\{-1, 0, 1, \omega\}$ to the indices on the bottom. This yields a total of 64 distinct tilings. 34 of these were shown already in Fig.~\ref{eq:zipperlemmaproof}; only numbers 9 and 10 cannot be included because $S$ does not contain the tile $D_{7}$. The remaining 30 are shown in Fig.~\ref{eq:zipperlemmaproof2}.   
		
		The zipper lemma follows from sequential application of the identity in Eq.~(\ref{eq:zippercorr}).
	\end{proof}
	
	\begin{figure*}
		\includegraphics[width=1 \linewidth]{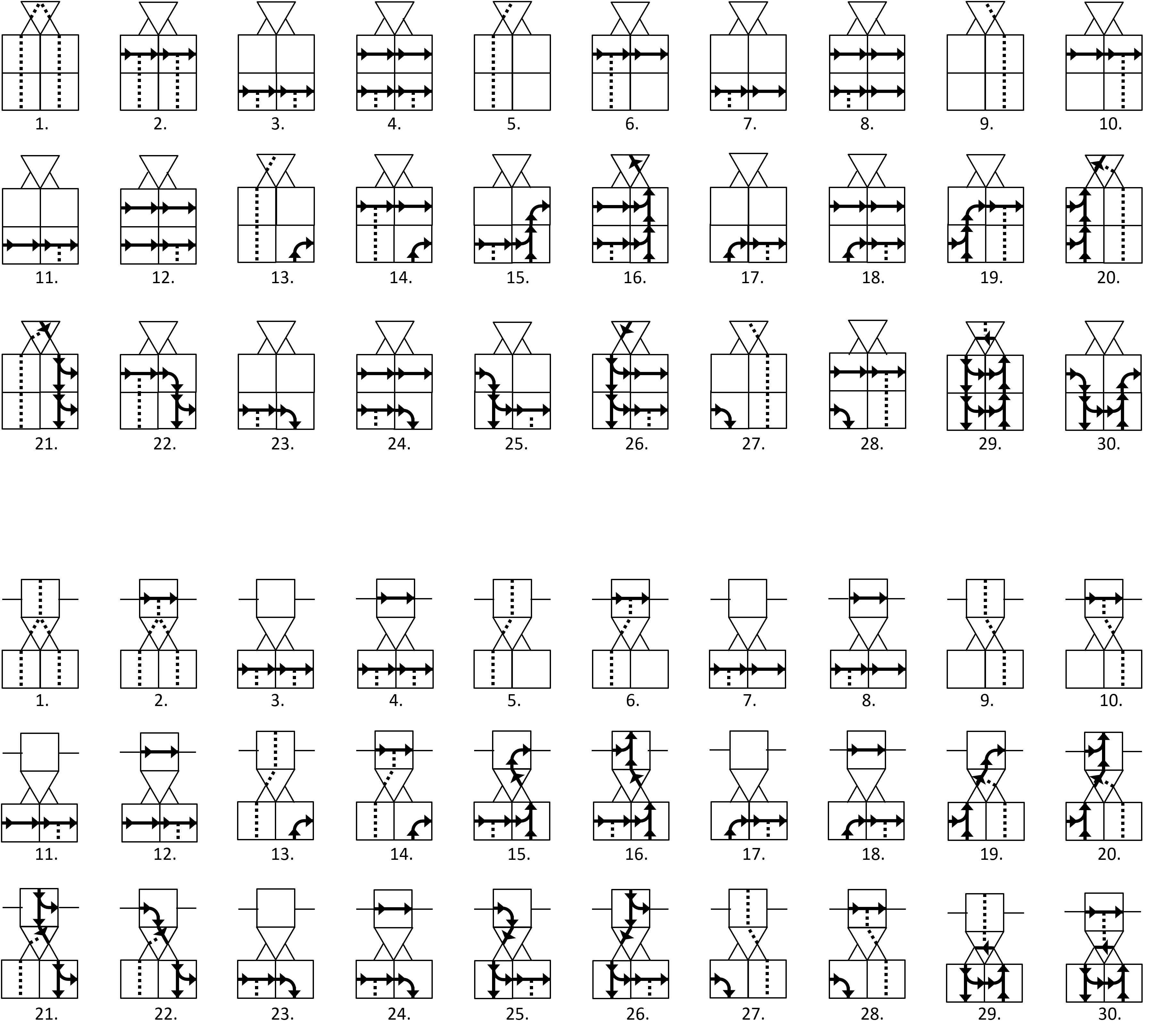}
		\caption{A subset of the tilings of the networks in Eq.~(\ref{eq:zippercorr}) using the $C$ and $S$ tensors.}\label{eq:zipperlemmaproof2}
	\end{figure*}
	
	Using the zipper lemma, we can pull a horizontal layer of triangular tensors in Fig.~{\ref{fig:TNequivalence2}(b)} downwards through the tensor network. At each step, two $C$ tensors are merged into one. Pulling down each triangular tensor is analogous to closing a zipper. The end result is shown in  Fig.~{\ref{fig:TNequivalence2}(c)}.
	
	The final step to proving equivalence between the original binary-height tensor network Fig.~{\ref{fig:TN2defs}(a)} and the height renormalization tensor network (Fig.~{\ref{fig:rntndefs2}(a)} is to note that the top $C$ tensor from Fig.~{\ref{fig:TNequivalence2}(c)} can be removed, since
	\begin{align}
		\includegraphics[width=0.3 \linewidth]{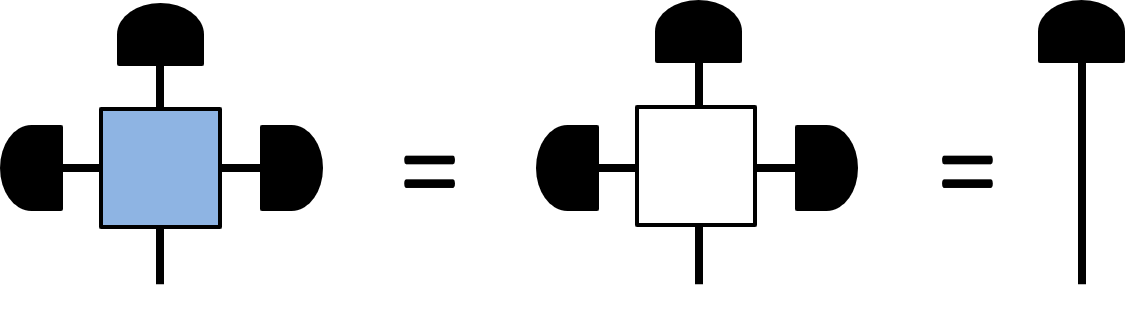}\label{eq:capremoval2}
	\end{align}
	where the square tile in the middle network is $A_4$.

	\section{Area weighted walks}\label{sec:TNtgen}
	The ground state of the $t>0$ area-weighted case (with open boundary conditions) is
	\begin{align} |\text{GS (t)}\rangle = \frac{1}{\mathcal{N}}\sum_{\substack{w \in \{ \text{Motzkin}\ \text{walks}\}}} t^{\mathcal{A}(w)} |w\rangle
		\label{cgs:Supp}. \end{align} 
	Note that $\mathcal{A}(w)$, the area under the walk $w$, is merely the sum of the heights at the vertices connecting adjacent walk segments, i.e, 
	\begin{align}
		\mathcal{A}(w) = \sum_{k=1}^{2n} h_{w}(k)
	\end{align}
	where $h_{w}(k)$ is the height of $w$ between the $(k-1)^{\text{th}}$ and $k^{\text{th}}$ segments. These heights can be defined via the following eigenvalue equation 
	\begin{align}
		\mathcal{H}_{k} \ket{w} = h_{w} (k) \ket{w}
	\end{align}
	where
	\begin{align}
		\mathcal{H}_{k} = \sum_{m=1}^{k} S_{Z, m},
	\end{align}
	and where $S_{Z,m}$ is spin-$z$ operator on the $m^{\text{th}}$ spin. 
	
	Thus,
	\begin{align}
		\ket{\text{GS}(t)} &= \prod_{k=1}^{2n} e^{\mathcal{H}_k} \ket{\text{GS}(1)}\\
		&=\prod_{m=1}^{2n} t^{(2n-m)S_{Z, m}} \ket{\text{GS}(1)}
	\end{align}
	Below we will show how operators of the form $t^{\alpha S_{Z,m}}$ can be absorbed into the binary height and renormalization tensor networks by redefining their building blocks.

	\subsubsection{Binary height network}
	We will denote
	\begin{align}
		\begin{tikzpicture}
			\begin{scope} 
				\node {$\displaystyle t^{\alpha S_{Z}} = $};
			\end{scope}
			\begin{scope}[xshift = 1.3cm]
				\node {\includegraphics[width=0.05\linewidth]{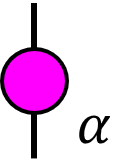}};
			\end{scope}
		\end{tikzpicture} \label{eq:triangletensor:Supp}
	\end{align}
	Then, by appending such tensors to the physical indices of the binary height network, we can represent the area weighted case as shown below for the $2n=8$ spin example.
	\begin{align}
		\includegraphics[width=0.45\linewidth]{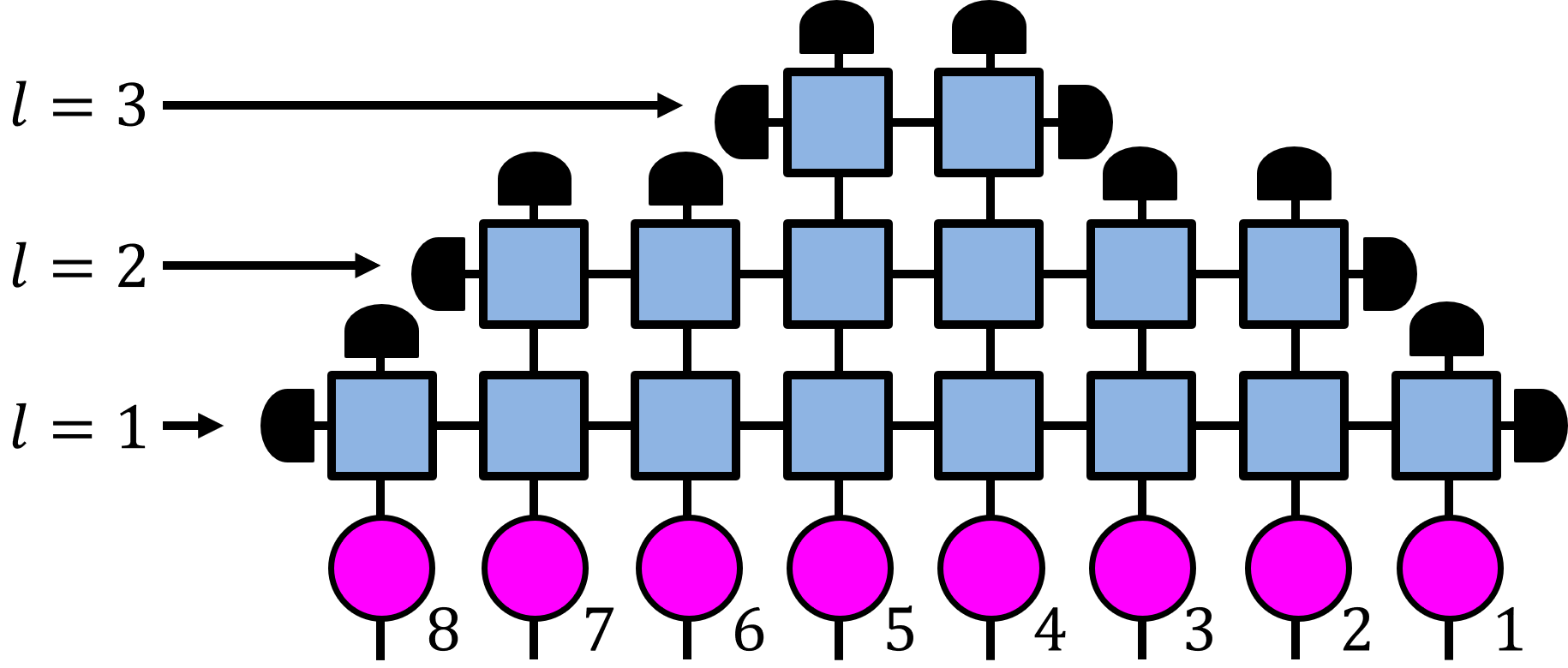}\label{eq:binheightweight} 
	\end{align}
	Note that the square tensors have the following symmetry with respect to the $S_{Z}$ operator:
	\begin{align}
		\includegraphics[width=0.5\linewidth]{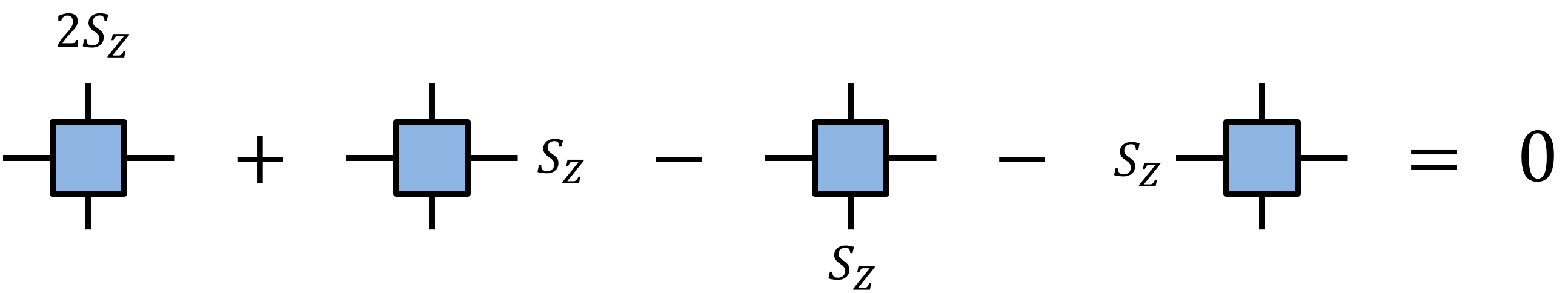}\label{eq:sqsymrel}
	\end{align}
	(in fact, each $A_i$ tile individually satisfies this relation).
	
	Using this identity, we can push the circle tensors onto the virtual legs one layer at at time
	\begin{align}
		\includegraphics[width=0.6\linewidth]{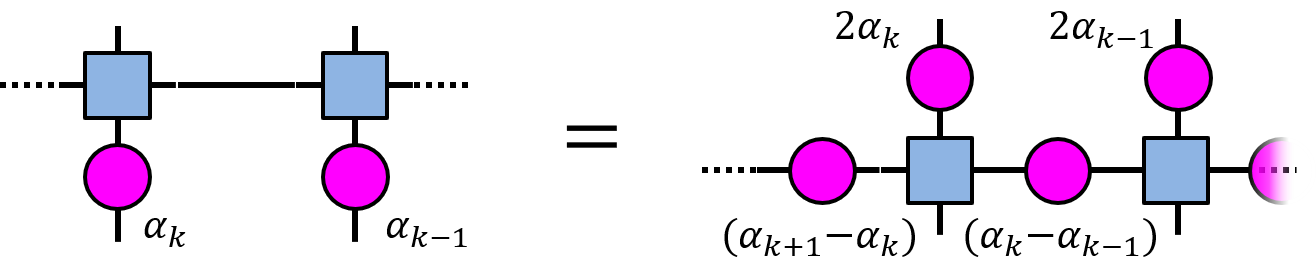}\label{eq:squpushthru}
	\end{align}
	Once the circle tensors reach the top boundary of $\ket{0}$ vectors, they vanish. All  remaining circle tensors reside on horizontal virtual indices and have $\alpha=2^{l-1}$, where the $l$-index indicates which row they appear in (see Eq.~(\ref{eq:binheightweight})).Next, by taking the square root, we split each such circle tensor into two with $\alpha=2^{l-2}$. Now every square tensor has a unique pair of adjacent circle tensors. The effect of these is to reweigh the $A_i$ tiles that are summed by each square, i.e., all square tensors on layer $l$ are modified so that
	\begin{align}
		\includegraphics[width=\linewidth]{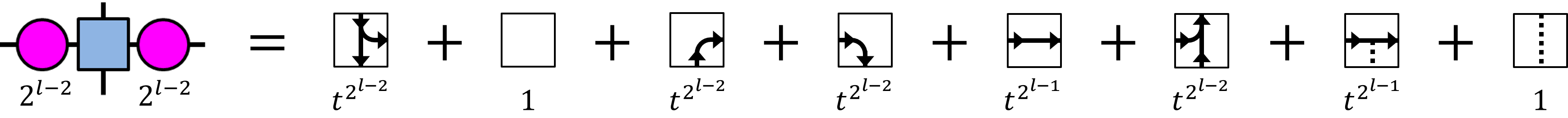} \label{eq:squredef}
	\end{align}
	Note, that each tile picks up a factor of $t^{2l-2}$ per horizontal arrow segment. Also, the final two tiles have been included in order to later generalize to the renormalization tensor network. Eq.~(\ref{eq:squredef}) remains valid in this case because we define $S_{Z}\ket{\omega}=0$. For the purposes of describing only the binary height network, the final two tiles in Eq.~(\ref{eq:squredef}) can be ignored.
	
	Thus, the binary height tensor network can be modified to include non-unit values of $t$ by introducing a layer-dependent reweighing of the tiles. Next we will generalize this procedure to the renormalization tensor network.
	
	\subsubsection{Renormalization tensor network}
	We can append circle tensors to the bottom of the renormalization tensor network to represent the area-weighted ground state. Below is an example for $2n=8$ spins.
	\begin{align}
		\includegraphics[width=0.55\linewidth]{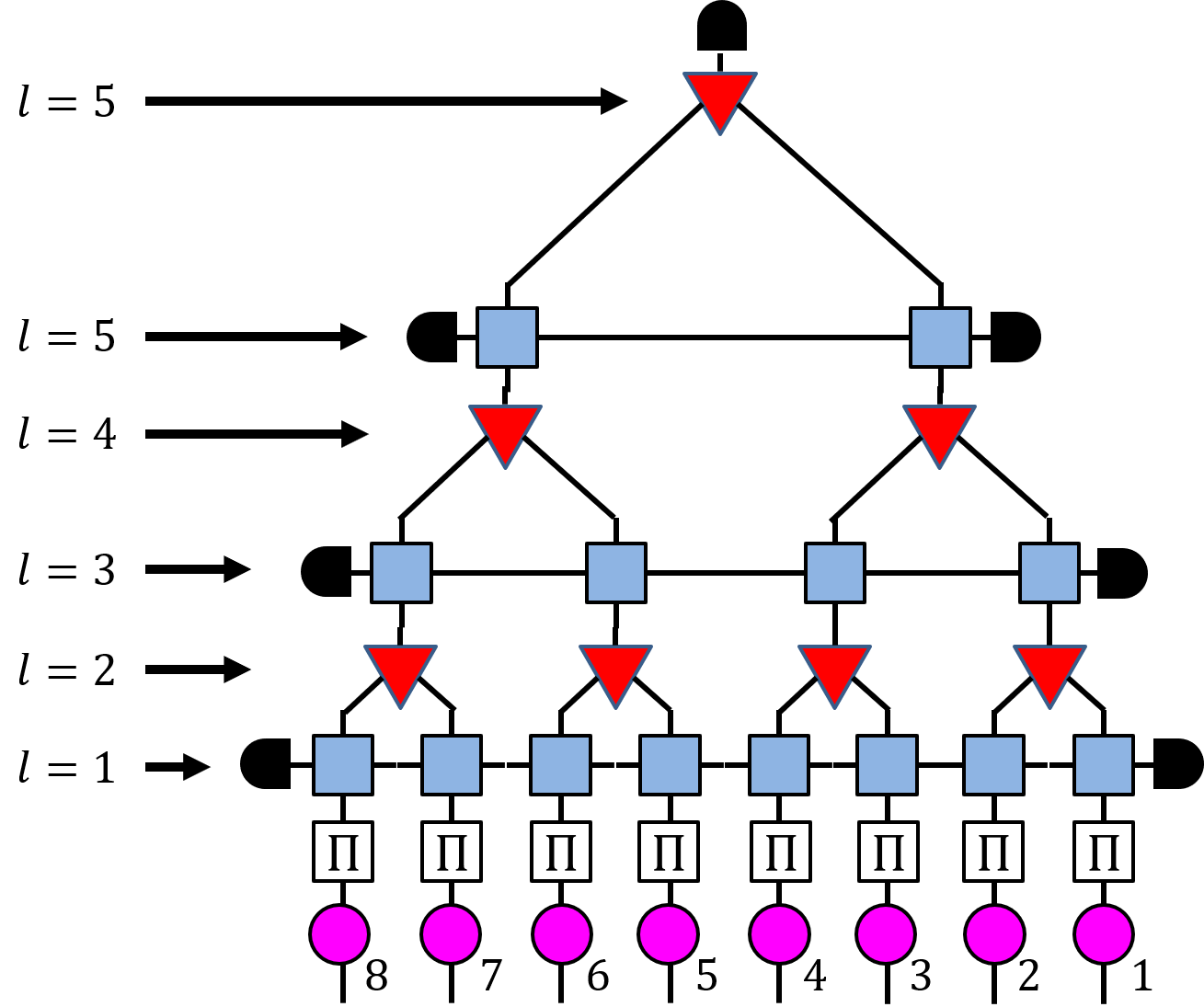} \label{eq:rnweight}
	\end{align}
	The circle tensors can be pushed upwards through the projectors $\Pi$ because they commute. Then, they can be pushed through a layer of square tensors using Eq.~(\ref{eq:squpushthru}) in the same way as described above. 
	
	The triangle tensors satisfy the following symmetry relation 
	\begin{align}
		\includegraphics[width=0.4\linewidth]{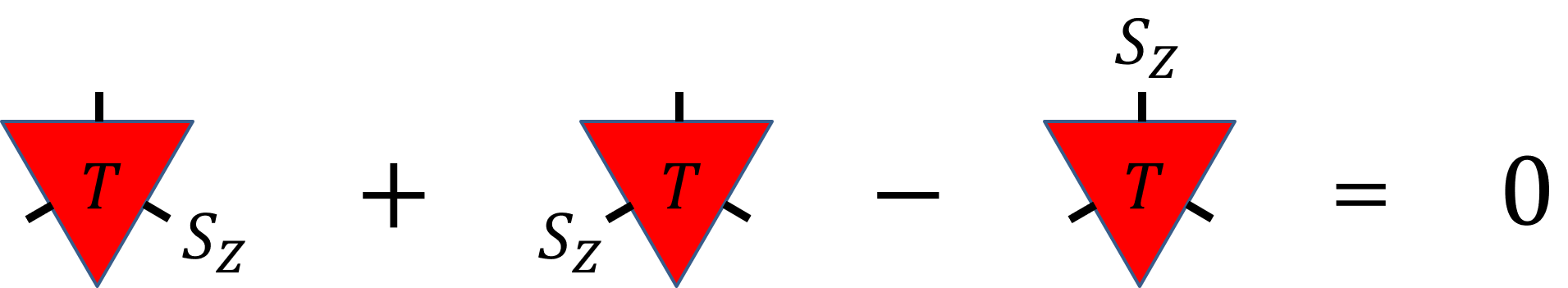}
	\end{align}
	(this relation also holds true for the each individual tile, $E_i$). From this, we can derive a ``push-through'' relation for the circle tensors: 
	\begin{align}
		\includegraphics[width=0.4\linewidth]{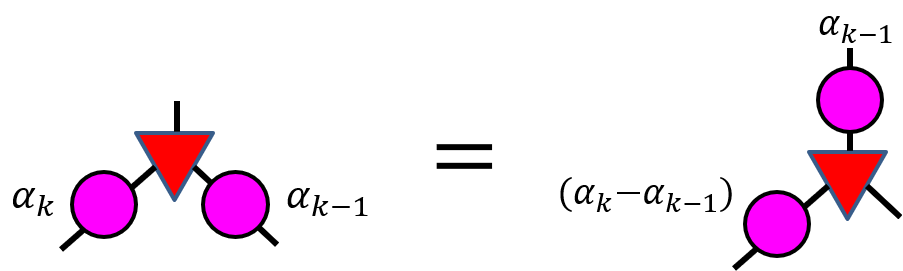}\label{eq:tripushthru}
	\end{align}
	We use relations Eq.~(\ref{eq:squpushthru}) and (\ref{eq:tripushthru}) one layer at a time until the circle tensors reach the top boundary contraction with $\ket{0}$, where they disappear. The remaining circle tensors reside on horizontal legs of the square tensors and the bottom-left legs of the triangle tensors. The former can be dealt with by reweighing square tiles as in Eq.~(\ref{eq:squredef}).
	
	The latter can be encorporated into the network by a layer-dependent reweighing of the $E_{i}$ tiles, where $l$ is as shown in Eq.~(\ref{eq:rnweight})
	\begin{align}
		\includegraphics[width=0.7\linewidth]{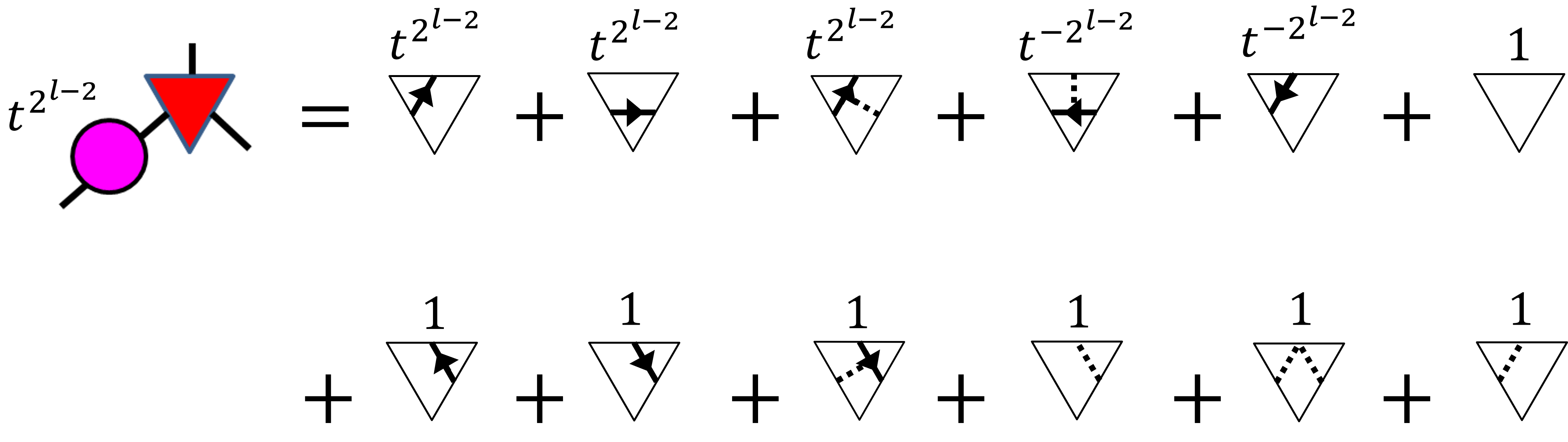}
	\end{align}
	
	Thus, the renormalization tensor network can also be modified to include non-unit values of $t$ by introducing a layer-dependent reweighing of the tiles. 
	
	
	\bibliography{Quantum_version.bbl}
\end{document}